\documentclass[11pt]{article}
\usepackage[utf8]{inputenc}
\usepackage[T1]{fontenc}
\usepackage[margin=1in]{geometry}
\usepackage{amsmath,amsthm,amssymb,thmtools,thm-restate,booktabs,color,doi,graphicx,latexsym,url,xcolor,xspace}
\usepackage[numbers,sort&compress]{natbib}
\usepackage{microtype,hyperref}
\usepackage[inline]{enumitem}
\setlist[itemize]{label=--}
\setlist[enumerate]{label=(\arabic*),labelindent=\parindent,leftmargin=*}

\usepackage{sectsty}
\allsectionsfont{\boldmath}

\definecolor{citecolor}{HTML}{0000C0}
\definecolor{urlcolor}{HTML}{000080}

\newtheorem{theorem}{Theorem}
\newtheorem{lemma}[theorem]{Lemma}
\newtheorem{corollary}[theorem]{Corollary}

\theoremstyle{remark}
\newtheorem{remark}[theorem]{Remark}

\newcommand{\namedref}[2]{\hyperref[#2]{#1~\ref*{#2}}}
\newcommand{\sectionref}[1]{\namedref{Section}{#1}}
\newcommand{\figureref}[1]{\namedref{Figure}{#1}}

\newcommand{\equationref}[1]{\hyperref[#1]{Eq~(\ref*{#1})}}
\newcommand{\theoremref}[1]{\hyperref[#1]{Theorem~\ref*{#1}}}
\newcommand{\lemmaref}[1]{\hyperref[#1]{Lemma~\ref*{#1}}}
\newcommand{\noteref}[1]{\hyperref[#1]{note~\ref*{#1}}}
\newcommand{\appendixref}[1]{\hyperref[#1]{Appendix~\ref*{#1}}}
\newcommand{\corollaryref}[1]{\hyperref[#1]{Corollary~\ref*{#1}}}

\renewcommand{\vec}[1]{\mathbf{#1}}

\DeclareMathOperator*{\E}{E}
\DeclareMathOperator*{\poly}{poly}
\DeclareMathOperator*{\polylog}{polylog}
\DeclareMathOperator*{\diam}{diam}
\DeclareMathOperator*{\outdeg}{out}
\DeclareMathOperator*{\indeg}{in}

\newcommand{\ppmodel}{\mathsf{PP}}

\newcommand{\taumix}{\tau_\textrm{mix}}

\newcommand{\id}{\operatorname{id}}
\newcommand{\tvnorm}[1]{\left\| #1 \right\|_{\operatorname{TV}}}

\newenvironment{myabstract}
{\list{}{\listparindent 1.5em%
		\itemindent    \listparindent
		\leftmargin    1cm
		\rightmargin   1cm
		\parsep        0pt}%
	\item\relax}
{\endlist}

\newenvironment{mycover}
{\list{}{\listparindent 0pt
		\itemindent    \listparindent
		\leftmargin    1cm
		\rightmargin   1cm
		\parsep        0pt}%
	\raggedright
	\item\relax}
{\endlist}

\newcommand{\myemail}[1]{\,$\cdot$\, {\small #1}}
\newcommand{\myaff}[1]{\,$\cdot$\, {\small #1}\par\medskip}

\hypersetup{
    colorlinks=true,
    linkcolor=black,
    citecolor=citecolor,
    filecolor=black,
    urlcolor=urlcolor,
    pdftitle={Fast Graphical Population Protocols},
    pdfauthor={Dan Alistarh, Rati Gelashvili and Joel Rybicki}
}

\title{Fast Graphical Population Protocols}

\begin{document}

\begin{mycover}
	{\huge\bfseries\boldmath Fast Graphical Population Protocols\par}

	\bigskip
	\bigskip
  \bigskip
  \textbf{Dan Alistarh}
    \myemail{dan.alistarh@ist.ac.at}
    \myaff{IST Austria}
    
    \textbf{Rati Gelashvili}
    \myemail{gelash@cs.toronto.edu}
    \myaff{University of Toronto}

    \textbf{Joel Rybicki}
    \myemail{joel.rybicki@ist.ac.at}
    \myaff{IST Austria}
    \bigskip
\end{mycover}

\medskip
\begin{myabstract}
  \noindent\textbf{Abstract.}
      Let $G$ be a graph on $n$ nodes.  In the stochastic population protocol model, a collection of $n$ indistinguishable, resource-limited nodes collectively solve tasks via pairwise interactions. In each interaction, two randomly chosen neighbors first read each other's states, and then update their local states. 
    A rich line of research has established tight upper and lower bounds on the complexity of fundamental tasks, such as majority and leader election, in this model, when $G$ is a \emph{clique}. Specifically, in the clique, these tasks can be solved \emph{fast}, i.e., in $n \operatorname{polylog} n$ pairwise interactions, with high probability, using at most $\operatorname{polylog} n$ states per node.

  In this work, we consider the more general setting where $G$ is an arbitrary graph, and present a technique for simulating protocols designed for fully-connected networks in any connected regular graph. Our main result is a simulation that is \emph{efficient} on many interesting graph families: roughly, the simulation overhead is polylogarithmic in the number of nodes, and quadratic in the conductance of the graph. 
  As a sample application, we show that, in any regular graph with conductance $\varphi$, both leader election and exact majority can be solved in $\varphi^{-2}  \cdot n \operatorname{polylog} n$ pairwise interactions, with high probability, using at most $\varphi^{-2} \cdot \operatorname{polylog} n$ states per node. 
  This shows that there are fast and space-efficient population protocols for leader election and exact majority on graphs with good expansion properties. 
  We believe our results will prove generally useful, as they allow efficient technology transfer between the well-mixed (clique) case, and the under-explored spatial setting. 
\end{myabstract}

\thispagestyle{empty}
\setcounter{page}{0}
\newpage

\section{Introduction}

Since the early days of computer science, there has been significant interest in developing an algorithmic theory of molecular and biological systems~\citep{turing1990chemical}. In distributed computing, \emph{population protocols}~\citep{angluin2006computation} have become a popular model for investigating the collective computational power of large collections of communication-bounded agents
with limited computational capabilities. This model consists of $n$ identical agents, seen as finite state machines, and computation proceeds via pairwise interactions of the agents, which trigger local state transitions. The sequence of interactions is provided by a scheduler, which picks pairs of agents to interact. 
Upon every interaction, the selected agents observe each other's states, and then update their local states. The goal is to have the system reach a configuration satisfying a given predicate, while minimising the number of interactions (time complexity) and the number of states per node (space complexity) required by the protocol.

Early work on population protocols focused on the computational power of the model, i.e., the class of predicates which can computed by population protocols under various interaction graphs~\citep{angluin2006computation,AAER07}.
More recently, the focus has shifted to understanding complexity thresholds, often in the form of fundamental  complexity trade-offs between time and space complexity, e.g.~\citep{angluin2008fast-computation,alistarh2015-fast,doty2018stable,gasieniec2018fast, alistarh2018space-optimal,BEFKKR18,berenbrink2020optimal,gasieniec2020time}; for recent surveys please see~\citep{elsaesser-survey, alistarh-survey}.

This line of work almost exclusively focuses on the \emph{uniform} stochastic scheduler, where each interaction pair is chosen uniformly at random \emph{among all pairs} of agents in the population, and the time complexity of a protocol is measured by the number of interactions needed to solve a task.
This is analogous to having a large \emph{well-mixed} solution of interacting particles, an assumption often used for modelling chemical reactions. 
However, many natural systems exhibit spatial structure and this structure can significantly influence the system dynamics. 

Indeed, there is a separation in terms of computational power for population protocols in the clique versus other interaction graphs: connected interaction graphs can simulate adversarial interactions on the clique graph by shuffling the states of the nodes~\citep{angluin2006computation} and population protocols on some interaction graphs can compute a strictly larger set of predicates than protocols on the clique; see e.g.~\citep{aspnes2009introduction} for a survey of computability results.

By comparison, surprisingly little is known about the \emph{complexity} of basic tasks in general interaction graphs under the stochastic scheduler. So far, only a handful of protocols have been analysed on general graphs. Existing analyses tend to be complex, and specialised to specific algorithms on limited graph classes~\citep{draief2012-convergence,cooper2016-fast,MNRS14,mertzios2017determining,BFKMW16}. This is natural:  given the intricate dependencies which arise due to the underlying graph structure, the design and analysis of protocols in the spatial setting is  understood to be challenging.

\subsection{Contributions}
In this work,
we provide a general approach showing that standard problems in population protocols can be solved \emph{efficiently} under \emph{graphical} stochastic schedulers, by leveraging solutions designed for complete graphs.  
Our results are as follows:
\begin{enumerate}

\item We give a general framework for simulating a large class of \emph{synchronous} protocols designed for \emph{fully-connected networks}, in the graphical stochastic population protocol model (see \figureref{fig:pp-model}). Thus, the user can design efficient (and simple to analyse) synchronous algorithms on a clique model, and transport the analysis automatically to the population protocol model on a large class of interaction graphs. For instance, on any  $d$-regular graph with edge expansion $\beta > 0$, the resulting overhead in parallel time and state complexity is in the order of $(d/\beta)^2 \cdot \polylog n$. 

 \item As concrete applications, we show that for any $d$-regular graph with edge expansion $\beta > 0$, there exist protocols for leader election and exact majority that stabilise both in expectation and with high probability\footnote{
The phrase ``with high probability'' (w.h.p.) means that we can choose constants so that the probability that the protocol fails to stabilise is at most $1/n^\lambda$ for any given constant $\lambda > 0$.} in $(d/\beta)^2 \cdot \polylog n$ parallel time, using $(d/\beta)^2 \cdot \polylog n$ states. 

 \item To complement the results following from the simulation, we also show that, on any graph $G$ with diameter $\diam(G)$ and $m$ edges, leader election can be solved both in expectation and with high probability in $O(\diam(G) \cdot m n^2 \log n)$ parallel time, using a constant-state protocol. This result provides the first running time analysis of the protocol of~\cite{beauquier2013self}.
\end{enumerate}

\begin{figure}[t]
  \centering
  \includegraphics[page=3,width=\textwidth]{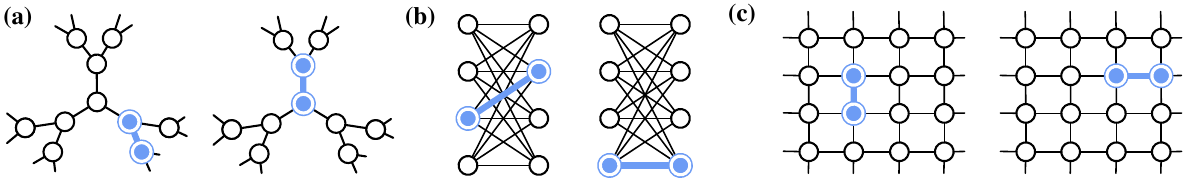}
  \caption{The graphical population protocol model. In each step, a random edge $\{u,v\}$ is selected and the nodes $u$ and $v$ interact (blue nodes). Examples of graph classes covered by our construction: (a) regular high-girth expanders, (b) bipartite complete graphs, (c) toroidal grids. \label{fig:pp-model}}
\end{figure}

\subsection{Technical overview}

Our reduction framework combines several techniques from different areas, and can be distilled down to the following ingredients.

We start by defining a simple \emph{synchronous, fully-connected} model of communication for the $n$ nodes, called the \emph{$k$-token shuffling model}.
This is the model in which the algorithm should be designed and analysed, and is similar, and in some ways simpler, relative to the standard population model. 
Specifically, nodes proceed in \emph{synchronous} rounds, in which every node $v$ first generates  $k$ tokens based on its current state. Tokens are then shuffled uniformly at random among the nodes. 
At the end of a round, every node $v$ updates its local state based on its current state, and the tokens it received in the round. \figureref{fig:token} illustrates the model.
This simple model is quite powerful, as it can simulate both \emph{pairwise} and \emph{one-way} interactions between all sets of agents, for well-chosen settings of the parameter~$k$.     

Our key technical result is that any algorithm specified in this round-synchronous $k$-token shuffling model can be \emph{efficiently} simulated in the graphical population model. 
Although intuitive, formally proving this result, and in particular obtaining bounds on the efficiency of the simulation, is non-trivial. 
First, to show that simulating \emph{a single round} of the $k$-token shuffling model can be done efficiently, we introduce new type of \emph{card shuffling process}~\cite{diaconis1993comparison, wilson2004mixing,caputo2010aldous,jonasson2012interchange}, which we call the $k$-stack interchange process, and analyse its mixing time by linking it to random walks on the symmetric group. 

Second, to allow correct and efficient asynchronous simulation of the synchronous token shuffling model, we introduce two new gadgets: (1)~a \emph{graphical} version of \emph{decentralised phase clocks}~\cite{alistarh2018space-optimal, gasieniec2018almost, gasieniec2018fast}, combined with (2)~an \emph{asynchronous} token shuffling protocol, which simulates the $k$-token interchange process in a graphical population protocol. 
The latter ingredient is our main technical result, as it requires both efficiently combining the above components, and carefully bounding the probability bias induced by simulating a synchronous model under asynchronous pairwise-random interactions. 

Finally, we instantiate this framework to solve exact majority and leader election in the graphical setting. We provide simple token-shuffling protocols for these problems, as well as backup protocols to ensure their correctness in all executions.

\begin{figure}[t]
  \centering
  \includegraphics[page=2,width=0.9\textwidth]{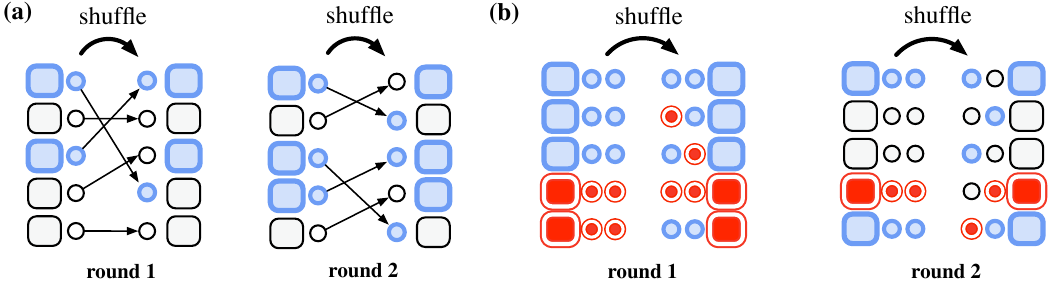}
  \caption{The synchronous $k$-token shuffling model with 5 nodes for $k=1$ and $k=2$. Rectangles are nodes and the small circles are tokens. In each round, nodes generate $k$ tokens based on their current state. Then all $nk$ tokens are shuffled randomly. After this, nodes update their state based on the vector of $k$ tokens they hold. (a) An execution of a protocol in the 1-token shuffling model. The arrows between tokens represent the random permutation used to shuffle tokens. (b) An execution of a protocol for $k=2$. Each node sends and receives two tokens. \label{fig:token}}
\end{figure}

\subsection{Implications}
Our results imply new and improved upper bounds on the time and state complexity of  majority and leader election for a wide range of graph families. In some cases, they improve upon the best known upper bounds for these problems. Please see Table~\ref{table:comparison} for a systematic comparison. Specifically, our results show that:
\begin{itemize}
\item In \emph{sparse} graphs with good expansion properties, such as constant-degree graphs with constant edge expansion (\figureref{fig:pp-model}a), our simulation has polylogarithmic time and state complexity overhead, relative to clique-based algorithms. Thus, good expanders admit fast protocols using polylogarithmic states, despite being sparser than the~clique.

\item In \emph{dense} graphs, we obtain similar bounds whenever $d/\beta \in \polylog n$ holds. This is the case for instance in $d$-dimensional hypercubes with $n=2^d$ nodes,
  but also in highly-dense clique-like graphs, such as regular complete multipartite graphs (\figureref{fig:pp-model}b), where the degree and expansion are both $\Theta(n)$.

\item In $D$-dimensional toroidal grids, we get algorithms with $n^{2/D} \polylog n$ parallel time and state complexity. These graphs include cycles (1-dimensional toroidal grids), two-dimensional grids (\figureref{fig:pp-model}c), three-dimensional lattices, and so on.
\end{itemize}
While our protocols guarantee \emph{fast} stabilisation in regular graphs with high expansion, they will stabilise in polynomial expected time in \emph{any connected graph}. The results can be carried over to certain classes of \emph{non-regular} graphs provided that they are not highly irregular and have high expansion; we discuss this in Section~\ref{sec:conclusions}, and provide examples in Appendix~\ref{app:non-regular}.

\begin{table}[t]
  \centering
  \small
  \begin{tabular}{@{}llllll@{}}
\toprule
Graphs & Task & States  & Parallel time  & Ref. & Note \\ \midrule
cliques           & EM     & 4          & $O(n \log n)$ &     \cite{draief2012-convergence} &  $\Omega(n)$ parallel time necessary~\cite{alistarh2017time}.  \\ 
& EM     & $O(\log n)$ & $\Theta(\log n)$        &  \cite{doty2020majority} & Optimal for certain protocols~\cite{alistarh2018space-optimal}. \\
     & LE      & $2$          & $\Theta(n)$ &   \cite{doty2018stable} & Optimal $O(1)$-state protocol.  \\
     & LE      &  $\Theta(\log \log n)$ & $\Theta(\log n)$        & \cite{berenbrink2020optimal}  &  Lower bounds in~\cite{alistarh2017time,sudo2020leader}. \\
\midrule
connected & EM     &   $4$            &     $\poly(n)$         &  \cite{draief2012-convergence,BFKMW16} & Various bounds (*) \\
  & LE     &   $6$              &   $O(\diam(G) \cdot  m n^2 \log n)$              &  {\bf new}  & Complexity analysis of \cite{beauquier2013self}. \\
           \midrule 

$d$-regular       & EM &   $(d/\beta)^2 \cdot \polylog n$     & $(d/\beta)^2 \cdot \polylog n$                             &    {\bf new} &  Also stabilises in non-reg.\ graphs. \\
       & LE &   $(d/\beta)^2 \cdot \polylog n$     & $(d/\beta)^2 \cdot \polylog n$                             &    {\bf new} & Also stabilises in non-reg.\ graphs. \\           
           
       \bottomrule
\end{tabular}
\caption{Protocols for exact majority (EM) and leader election (LE) for different graph classes. The state complexity is the number of states used by the protocol. The parallel time column gives the expected parallel time (expected number of steps divided by $n$) to stabilise. (*) In~\cite{draief2012-convergence}, the running time of the protocol is bounded by the initial discrepancy in the inputs and the spectral properties of the contact rate matrix; bounds in terms of $n$ are only given for select graph classes (paths, cycles, stars, random graphs and cliques). No sublinear in $n$ bounds on parallel time are given~in~\cite{draief2012-convergence}.}
\label{table:comparison}
\end{table}

It is known that, in the clique setting, constant-state protocols are necessarily slower than protocols with super-constant states~\cite{doty2018stable, alistarh2018space-optimal}. 
Our results suggest the existence of a similar complexity gap in the graphical setting. Specifically, on $d$-regular graphs with good expansion, such that $d/\beta \in \polylog n$, we provide      polylogarithmic-time protocols for both leader election and exact majority. 
This opens a significant complexity gap relative to known constant-state protocols on graphs. 
For instance, the 4-state exact majority protocol for general graphs~\cite{draief2012-convergence} requires $\Omega(n)$ parallel time even in regular graphs with high expansion, if node degrees are $\Theta(n)$. (A simple example is the complete bipartite graph given in \figureref{fig:pp-model}b.) Yet, our protocols guarantee stabilisation in only $\polylog n$ parallel time in both low and high degree graphs, as long as $d/\beta$ is at most $\polylog n$. 

\subsection{Roadmap}

 We overview related work in \sectionref{sec:rw}. \sectionref{sec:prelim} defines the model and notation, while Sections \ref{sec:interchange} to \ref{sec:applications} develop our framework, from shuffling processes, to the simulation, and  applications. \sectionref{sec:constant-le} gives an analysis for a constant-state protocol for leader election that stabilises in polynomial expected time in any connected graph.
We conclude in \sectionref{sec:conclusions} by discussing some open problems.

\section{Related Work}
\label{sec:rw}

\paragraph{Computability for graphical population protocols.} 
A variant of the graphical setting was already considered in the foundational work of Angluin et al.~\cite{angluin2006computation}, which also uses a state shuffling approach. However, the resulting line of work focused on \emph{computational power} in the case where the number of states per node is constant~\cite{angluin2006computation,angluin2006stably, angluin2008self, AAER07, CMNPS11, blondin2018large}. 
A key difference is that we aim to simulate pairwise interactions under the uniform stochastic scheduler, as fast protocols in the clique require that pairwise interactions are uniformly random~\cite{elsaesser-survey,alistarh-survey}. Thus, one of the main technical challenges is to devise an \emph{efficient} shuffling procedure that guarantees that the simulated interactions are (almost) uniform.

In addition, self-stabilising population protocols on graphs have been investigated particularly in the context of leader election~\cite{angluin2008self,beauquier2013self,yokota2020time,chen2019self,chen2020ssle}.
While the problem is not always solvable on all graph families~\cite{angluin2008self}, Chen and Chen~\cite{chen2019self} gave a constant-state protocol for leader election with exponential stabilisation time in directed cycles and 2-dimensional toroidal grids. Later, they gave a protocol for $d$-regular graphs using $O(d^{12})$ states~\cite{chen2020ssle}. 

Beauquier, Blanchard and Burman~\cite{beauquier2013self} noted that without the requirement of self-stabilisation, 
 leader election can be solved on every connected graph by a constant-state protocol. We provide the first running-time upper bounds for this protocol here. 
Please see Table~\ref{table:comparison} for additional references, and bound comparison.  

\paragraph{Complexity in the clique model.}  A parallel line of work has focused on determining the fundamental space-time trade-offs for key tasks, such as majority and leader election, when the interaction graph is a \emph{clique}~\cite{doty2018stable,draief2012-convergence,MNRS14,alistarh2017time,alistarh2018space-optimal, BKKO18,BEFKKR18,berenbrink2020optimal,gasieniec2020time}. In this case, tight or almost-tight  complexity trade-offs are now known for these problems~\cite{berenbrink2020optimal, gasieniec2020time, doty2020majority, alistarh2018space-optimal}.

The vast majority of the work on complexity has focused on the clique case~\cite{elsaesser-survey,alistarh-survey}. Two natural justifications for this choice are that: (1)~the clique is a good approximation for well-mixed solutions, and (2)~the analysis of population protocols can be difficult enough even without additional complications due to graph structure.
Bounds on non-complete graphs have been studied for exact~\cite{draief2012-convergence} and approximate majority~\cite{MNRS14,mertzios2017determining}, with some recent work considering \emph{plurality consensus}~\cite{cooper2013coalescing,cooper2016-fast, BFKMW16} in a related model. 
The recent survey of~\cite{elsaesser-survey} points out that running time on general graphs is poorly understood, and sets this as an open question. We take a first step towards addressing this gap.

\paragraph{Interacting particle systems.}
Another related line of work investigated dynamics of interacting particle systems on graphs, e.g.~\cite{aldous-fill-2014}. However, in this context dynamics are often assumed to be round-synchronous,  
which allows the use of more powerful techniques, related to independent random walks on graphs~\cite{lovasz1993random}. 
Cooper, Els\"asser, Ono and Radzik~\cite{cooper2013coalescing} analysed the coalescence time of independent random walks on a graph in terms of the expansion properties of the graph, where each node initially holds a unique particle, and in each step particles randomly move to another node. Whenever, two particles meet, they coalesce into a single one, which continues its walk. 
We also employ token-based protocols on graphs, but in our case tokens are shuffled between nodes instead of coalescing. 

Token-based processes have also been used to implement efficient, randomised rumour spreading protocols. For example, Berenbrink, Giakkoupis and Kling~\cite{berenbrink2018tight} analysed the cover time of a synchronous coalescing-branching random walk on regular graphs. Similarly to our work, they use conductance to bound the behaviour of this process in regular graphs.
In this work, we use token-based population protocols on graphs, where the tokens are shuffled between nodes during an interaction and the tokens instead of coalescing, may also interact in other ways. 

\paragraph{Plurality consensus on expanders.} In plurality consensus, there are $k>1$ opinions and the task is the agree on opinion supported by the most nodes. Berenbrink, Friedetzky, Kling, Mallmann-Trenn and Wastell~\cite{BFKMW16} present a protocol for the plurality consensus problem in a synchronous pull-based interaction model. Their protocol also circulates tokens, and samples their count periodically (after mixing) to estimate opinion counts, running into the issue that the token movements are correlated. The authors provide a generalisation of a result by Sauerwald and Sun~\cite{SS} in order to show that the joint token distribution is negatively correlated, and therefore the token counting mechanism~concentrates.

In this work, we also employ a token exchange protocol, and encounter non-trivial correlation issues. 
However, we resolve these issues differently: we characterise the distribution of the token interactions using the $k$-stack interchange process, and bound its total variation distance relative to the uniform distribution, showing that the two distributions are indistinguishable in polynomial time with high probability.
More generally, the goal of our construction is different, as we aim to provide a general framework to efficiently simulate pairwise random node interactions.

\paragraph{Shuffling processes.} Our results also connect to the work on card shuffling processes, which have a long and rich history~\cite{diaconis1981generating,aldous1983random,diaconis1993comparison,wilson2004mixing,caputo2010aldous,dieker2010interlacings,jonasson2012interchange,oliveira2013mixing}. While many of these processes are simple to describe, they are often surprisingly challenging to analyse. Here, we focus on key results related to the interchange process, where the cards are placed on the nodes of a graph and shuffling is performed by randomly exchanging cards between adjacent nodes. We note that much of the work has aimed to identify sharp bounds on the mixing time for the interchange process on various graphs.
 
 Diaconis and Shahshahani~\cite{diaconis1981generating} gave sharp bounds of the order $\Theta(n \log n)$ on the mixing time of the random transpositions shuffle, i.e., interchange process on the clique. Aldous~\cite{aldous1983random} established that the mixing time of the interchange process on the path is bounded by $\Omega(n^3)$ and $O(n^3 \log n)$; later Wilson~\cite{wilson2004mixing} showed that the mixing time is in fact $\Theta(n^3 \log n)$. Diaconis and Saloff-Coste~\cite{diaconis1993comparison} developed a powerful technique for upper bounding the mixing time of a random walk on a finite group by comparing it to another walk with known behaviour via certain Dirichlet forms. Our analyses of the $k$-stack interchange process also rely on this comparison technique. 

 A decade later Wilson~\cite{wilson2004mixing} gave a general technique for proving lower bounds for many shuffling processes. In particular, he showed that the mixing time on the two-dimensional $\sqrt{n} \times \sqrt{n}$ grid is $\Theta(n^2 \log n)$ and $\Omega(n \log^2 n)$ on the hypercube. Subsequently, Jonasson~\cite{jonasson2012interchange} gave additional upper and lower bounds on the interchange process on various graphs, including showing that the mixing time on the hypercube and constant-degree expanders is at most $O(n \log^3 n)$ and $O(\rho m n\log n)$ on any $m$-edge graph with radius $\rho$.
For a further exposition to this area, we refer to
\cite{levin2017markov}.

In this work, we introduce and analyse a generalisation of the interchange process, called the $k$-stack interchange process, where each node holds $k > 0$ cards instead of one.

\paragraph{Clique emulation.}
The general idea of clique emulation over a general communication graph is a classic one, and has been used in other stronger models of distributed computing, e.g.~\cite{avin2017distributed, ghaffari2017distributed}. 
For example, Ghaffari, Kuhn and Su~\cite{ghaffari2017distributed} utilise parallel random walks on graphs to come up with an efficient permutation routing scheme for the synchronous CONGEST model, with running time bounded by the mixing time of the random walk on the graph. In contrast, we bound the running times of our asynchronous protocols using the mixing time of the $k$-stack interchange process.

\section{Preliminaries} 
\label{sec:prelim}

\paragraph{Graphs.}
A graph $G = (V,E)$ is $d$-regular if every node $v \in V$ is adjacent to exactly $d$ other nodes. 
The edge boundary of a set $S \subseteq V$ is the set $\partial S \subseteq E$ of edges with exactly one endpoint in~$S$. The \emph{edge expansion} of the graph $G$~is defined as 
\[
\beta = \min \left\{ \frac{|\partial S|}{|S|} : S \subseteq V, |S| \le n/2 \right\}.
\]
If $G$ is regular, its \emph{conductance} is $\beta/d$.
Unless otherwise mentioned, all graphs are assumed to be regular and connected.

\paragraph{Probability distributions.}
Let $E$ be a finite set. We say $\mu \colon E \to [0,1]$ is a probability distribution on $E$ if $\sum_{x \in E} \mu(x) = 1$ holds.
For $A \subseteq E$ we write $\mu(A) = \sum_{x \in A} \mu(x)$.
The \emph{uniform distribution} on $E$ is the distribution $\nu$ defined by $\nu(x) = 1/|E|$. The \emph{support} of $\mu$ is the set $\{ x : \mu(x) > 0 \}$.
The \emph{total variation distance} between distributions $\mu_1$ and $\mu_2$ on $E$~is
\[
\tvnorm{ \mu_1 - \mu_2} = \frac{1}{2} \sum_{x \in E} | \mu_1(x) - \mu_2(x) | = \max_{A \subseteq E} | \mu_1(A) - \mu_2(A)|.
\]
We say that $\mu$ is $\varepsilon$-uniform on $E$ if $\tvnorm{\mu - \nu} \le \varepsilon$.

\paragraph{Permutations and the symmetric group.}
Let $N > 0$ be a positive integer and $[N] = \{0, \ldots, N-1\}$.
A permutation on $[N]$ is a bijection from $[N]$ to $[N]$.
The symmetric group $S_N$ over $[N]$ is the group consisting of the set of all permutations on $[N]$ with function composition as the group operation and identity element $\id$ defined by $\id(i) = i$.
The inverse $x^{-1}$ of an element $x \in S_N$ is the map satisfying $x^{-1} \cdot x = x \cdot x^{-1} = \id$.
A \emph{transposition} $(i~j) \in S_N$ of $i$ and $j$ is the permutation that swaps the elements $i$ and $j$, but leaves other elements in place.
We say that a set $H \subseteq S_N$ generates $S_N$ if every element of $S_N$ can be expressed as a finite product of elements in $H$ and their inverses. We use $\cdot$ and $\circ$ interchangeably to denote function composition.

Let $\mu$ be a symmetric probability distribution on $S_N$, i.e., $\mu(x) = \mu(x^{-1})$. The \emph{random walk on $S_N$} with increment distribution $\mu$ is a discrete time Markov chain with state space $S_N$. In each step, a random element $x$ is sampled according $\mu$ and the chain moves from state $y$ to state $xy$.
Thus, the probability of transitioning from state $x$ to state $yx$ is $\mu(y)$.
The holding probability of the random walk is $\alpha = \mu(\id)$. The following remark summarises some useful properties of such random walks; see e.g.~\cite{levin2017markov} for proofs.

\begin{remark}
  Let $\mu$ be an increment distribution for a random walk on $S_N$. 
  \begin{enumerate}[noitemsep]
    \item The uniform distribution $\nu$ on $S_N$ is a stationary distribution for the random walk.
    \item The random walk is reversible if and only if $\mu$ is symmetric.
    \item The random walk is irreducible if and only if the support of $\mu$ generates $S_N$.
    \item If $\mu(\id) > 0$, then the random walk is aperiodic.
  \end{enumerate}
\end{remark}

\paragraph{Mixing times.}
Let $\nu$ be the uniform distribution on $S_N$ and be $p^{(t)}$ be the probability distribution over states of the chain after $t$ steps. 
Following~\cite{diaconis1993comparison}, we define the $\ell^s$-norm and the normalised $\ell^s$-distance to stationarity for $s > 0$ as:
\[
\left\| \mu \right\|_s = \left( \sum_{x} |\mu(x)|^s \right)^{1/s} \quad \textrm{ and } \quad d_s(t) = |S_N|^{1-1/s} \cdot \| p^{(t)} - \nu \|_s.
\]
The total variation distance and the normalised distances satisfy
$2 \tvnorm{ p^{(t)} - \nu} = d_1(t) \le d_2(t)$,
where the latter inequality follows from the Cauchy-Schwarz inequality. We define the $\varepsilon$-mixing time as
$ \tau(\varepsilon) = \min \{ t : d_1(t) \le 2\varepsilon \}$.
We refer to the value $\taumix=\tau(1/2)$ as the \emph{mixing time} of the walk.
Note that $\tau(\varepsilon) \le \lceil \log_2 \varepsilon^{-1} \rceil \cdot \taumix$.

\paragraph{Tasks.}
Let $\Sigma$ and $\Gamma$ be nonempty finite sets of input and output labels, respectively.
A task $\Pi$ on a set $V$ of $n$ nodes is a function $\Pi$ that maps any input labelling $z \colon V \to \Sigma$ to a set $\Pi(z) \subseteq \Gamma^V$ of feasible output labellings.
If $\Pi(z) = \emptyset$, then we say that $z$ is an infeasible input. We focus on two tasks:
\begin{itemize}
\item In leader election, the input is the constant function $z(v) = 1$ and the output labelling~$z'$ is feasible iff there exists $v \in V$ such that $z'(v) = 1$ and $z'(u) = 0$ for all $u \neq v$. 
That is, exactly one node should output 1 and all others should output 0.
\item In the majority task, the inputs are given by $z \colon V \to \{0,1\}$ and $z' \in \Pi(z)$ if $z'(v) = b$, where $b$ is the input value held by the majority of the nodes. As conventional, the input with equally many zeros and ones is taken to be infeasible.
\end{itemize}

\paragraph{Graphical stochastic population protocols.}
Let $G = (V,E)$ be a graph. In the graphical stochastic population model, abbreviated as $\ppmodel(G)$, the computation proceeds \emph{asynchronously}, where in each time step $t > 0$: 
\begin{enumerate}[noitemsep]
\item a stochastic scheduler picks uniformly at random a pair $e_t = (u,v)$ of neighbouring nodes, 
\item the nodes $u$ and $v$ read each other's states and update their local states.
\end{enumerate}
As is common in population protocols, we assume that the node pairs are \emph{ordered}, which will allow us to distinguish the two nodes: node $u$ is called the \emph{initiator} and $v$ is the \emph{responder}. 
We assume that nodes have access to independent and uniform random bits. Specifically, upon each interaction, both $u$ and $v$ are provided with a single random bit each. We note that this assumption is common in the context of population protocols, e.g.~\cite{gasieniec2018fast}, and can be justified practically by the fact that chemical reaction network (CRN) implementations can directly obtain random bits given the structure of their interactions~\cite{brijder2019computing}. 

Formally, a protocol for a task $\Pi$ is a tuple $\vec A = (f, \ell_\textrm{in}, \ell_\textrm{out})$, where $f \colon S \times \{0,1\} \times S \times \{0,1\} \to S \times S$ is the state transition function and $S$ is the set of states, $\ell_\textrm{in} \colon \Sigma \to S$ maps inputs to initial states, and $\ell_\textrm{out} \colon S \to \Gamma$ maps states to outputs.  A configuration is a map $x \colon V \to S$ and  $x_0 = \ell_\textrm{in} \circ z$ is the initial configuration on input $z$.
An asynchronous schedule is a random sequence $(e_t )_{t \ge 1}$ of the interaction pairs. An execution is the sequence $(x_t)_{t \ge 0}$ of configurations given by
\[
x_{t+1}(u), x_{t+1}(v) = f\left( x_t(u), q_{t+1}(u), x_t(v), q_{t+1}(v) \right) \textrm{ and } x_{t+1}(w) = x_t(w) \textrm{ for } w \in V \setminus \{u,v\},
\]
where $(u,v) = e_{t+1}$ and $q_{t+1}(u) \in \{0,1\}$ is the random bit provided to the node $u$ during the interaction. The output of the protocol at step $t$ is given by $z'_t = \ell_\textrm{out} \circ x_t$. 

We say that $\vec A$ stabilises on input $z$ by step $T$ if
$ z'_{t+1} = z'_t$ and $z'_t \in \Pi(z)$
 holds for all $t \ge T$. Moreover, $\vec A$ solves the task $\Pi$ with probability at least $p$ in $T(\vec A)$ steps if the protocol stabilises by step $T(\vec A)$ on any feasible input with probability at least $p$. The state complexity of the protocol is $S(\vec A) = |S|$, i.e., the number of states used by the protocol.

\paragraph{Synchronous token protocols.}
In the synchronous $k$-token shuffling model, we assume that there are $n$ agents which communicate in a round-based fashion using \emph{tokens}. In each round,
\begin{enumerate}[noitemsep]
\item every node $v$ generates exactly $k$ tokens based on its current state,
\item all $nk$ tokens are shuffled uniformly at random so that each node gets exactly $k$ tokens,
\item every node $v$ updates its local state based on its current state and the $k$ tokens it received.
\end{enumerate}
Let $X$ be the set of states a node can take and $Y$ be a set of distinct token types.
An algorithm in the token shuffling model is a tuple $\vec B = (f,g,\ell_\textrm{in}, \ell_{\textrm{out}})$. The map $f \colon X \times Y^k \to X$ is a state transition function, and $g \colon X \to Y^k$ determines which tokens each node creates at the start of each round. 
As before, $\ell_{\textrm{in}} \colon \Sigma \to X$ maps input values to initial states and $\ell_{\textrm{out}} \colon X \to \Gamma$ maps the state of a node onto an output value. The initial configuration on input $z$ is $x_0 = \ell_{\textrm{out}} \circ z$.

A \emph{synchronous schedule} is a sequence $(\sigma_r)_{r \ge 1}$, where the permutation $\sigma_r \in S_{nk}$ describes how the tokens are shuffled in round~$r$.
For any $y \colon [nk] \to Y$, we let $y(v_0, \ldots, v_{k-1}) = (y(v_0), \ldots, y(v_{k-1}))$. A synchronous execution induced by $(\sigma_r)_{r \ge 1 }$ on input $z$ 
is defined by
  \[
  y_{r+1}(v_0, \ldots, v_{k-1}) = (g \circ x_r)(v) \quad \textrm{ and } \quad  x_{r+1}(v) = f\left( x_r(v), \left(y_{r+1} \circ \sigma_{r+1}\right)\left(v_0, \ldots, v_{k-1}\right) \right),
  \]
  where 
  $y_{r}(v_0, \ldots, v_{k-1})$ and $(y_{r} \circ \sigma_{r+1})(v_0, \ldots, v_{k-1})$, respectively, are the $k$ tokens generated and received by node~$v$ during round $r$.
  
  We assume the uniform synchronous scheduler, which picks each permutation $\sigma_r$ independently and uniformly at random from the set of all permutations $S_{nk}$. The output of node $v$ at the end of round $r$ is $z'_r(v) = (\ell_{\textrm{out}} \circ x_r)(v)$.
  The synchronous algorithm $\vec B$ stabilises on input $z$ in $R$ rounds if $z_{r+1} = z'_r$ and $ z'_r \in \Pi(z)$ holds for all $r \ge R$. The algorithm solves the problem $\Pi$ if it stabilises in $R$ rounds on any feasible input with probability at least $p$.

\section{Shuffling on graphs: the \texorpdfstring{$k$}{k}-stack interchange process}\label{sec:interchange}

We now describe a shuffling process on graphs, which we call the \emph{$k$-stack interchange process}. This process will be useful in our analysis, and is a variant of the classic graph interchange process, e.g.~\cite{dieker2010interlacings, jonasson2012interchange}.  
We analyse its mixing time using the path comparison method of Diaconis and Saloff-Coste~\cite{diaconis1993comparison},  leveraging a classical flow result of Leighton and Rao~\cite{leighton1999multicommodity}.

\paragraph{\boldmath The $k$-stack interchange process.}
Let $G = (V,E)$ a graph with $n$ vertices $\{0, \ldots, n-1\}$ and $N = kn$ for $k > 0$.
Assume each node of $G$ holds a stack of exactly $k$ cards, and consider the shuffling process where, 
in every time step, one of the following actions is taken:
\begin{enumerate}[noitemsep]
\item with probability $1/2$, move the top card of a random node to the bottom of its stack,
\item with probability $1/4$, choose a random edge $\{u,v\}$ and swap the top cards of $u$ and $v$,
\item with probability $1/4$, do nothing.
\end{enumerate}
We refer to this process as the \emph{$k$-stack interchange process on $G$}. The special case of $k=1$ is the classic interchange process on $G$ with holding probability $3/4$, as the first rule does not do anything on stacks of size 1. For $k > 1$, the holding probability will be $1/4$. Instances of the process for $k=1$ and $k=2$ are illustrated in \figureref{fig:interchange}.

\begin{figure}[t]
\centering
  \includegraphics[page=4,width=0.85\textwidth]{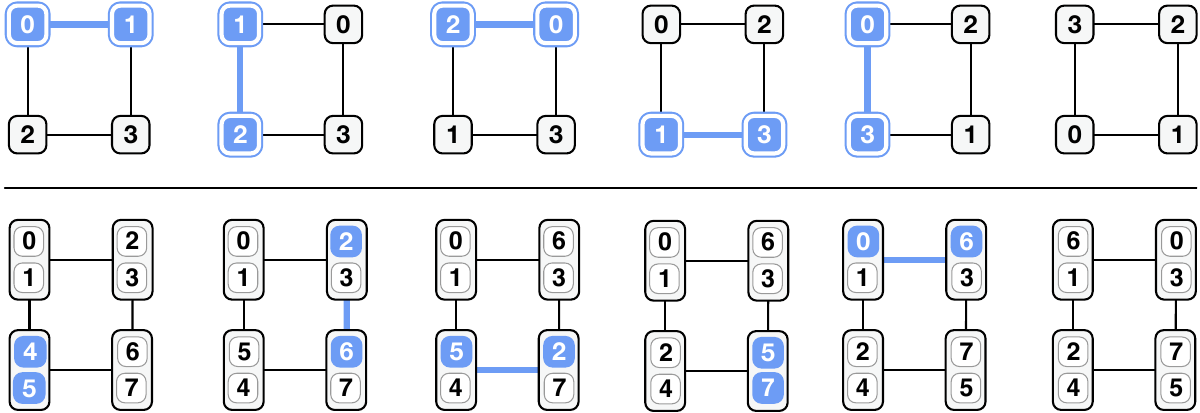}
  \caption{Interchange dynamics on a 4-cycle. In each step, blue cards are swapped. Top row: The 1-stack interchange process. Bottom row: The 2-stack interchange process. In each step, a randomly selected node either moves its top card to the bottom of its stack or exchanges it with the top card of a randomly selected neighbour.
    \label{fig:interchange}}
\end{figure}

\begin{restatable}{theorem}{mixtheorem}\label{thm:d-beta-mixing}
Let $G$ be a $d$-regular graph with edge expansion $\beta > 0$. For any constant $k>0$, the mixing time of the $k$-stack interchange process on $G$ is
$O\left( \left(d/\beta\right)^2 n \log^3 n \right)$.
\end{restatable}

We prove this theorem in \appendixref{apx:interchange}. In Section~\ref{sec:simulation}, we will show that this shuffling process can be implemented efficiently in the graphical population protocol model.

\section{Decentralised \emph{graphical} phase clocks}\label{sec:clocks}

We now describe a \emph{bounded phase clock} construction for the stochastic population protocol model over regular graphs. Interestingly, the construction can be generalised to \emph{non-regular} graphs, assuming that node degrees do not deviate too much from the average degree; see \appendixref{app:non-regular}. Our approach generalises that of Alistarh et al.~\cite{alistarh2018space-optimal}, who built a leaderless phase clock on cliques leveraging the classic two-choice load balancing process~\cite{azar1999balanced, peres2014graphical}.

\paragraph{Phase clocks.}
Let $\phi > 0$ be an integer and consider a population protocol $\vec C$ with state variables $c(v) \in \{0, \ldots, \phi-1\}$ for each $v \in V$. The variable $c(v)$ represents the value of the clock at node $v$.
Let $c(v,t)$ be the clock value node $v$ has at the end of time step $t$ (regardless of whether it was active during that step).
We define the distance $D$ between two clock values and the skew $\Delta$ of the clock at the end of step $t$, respectively, as follows: 
\[
D(x,y) = \min \{ |x-y|, \phi - |x-y| \} \quad \textrm{ and } \quad \Delta(t) = \max_{u,v \in V} D\left(c(u,t),c(v,t)\right).
\]
We say that the protocol $\vec C$ implements a $(\phi,\gamma,\kappa)$-clock if for all $t \ge 0$ the following hold:
\begin{enumerate}[noitemsep]
\item $\Pr[\Delta(t) \ge \gamma] < t/n^\kappa$, and
\item $c(v,t+1) = c(v,t) + 1 \bmod \phi$ for exactly one $v \in V$ and $c(u,t+1) = c(u,t)$ for all $u \in V \setminus \{v \}$.
\end{enumerate}
Intuitively, $\phi$ is the length of a phase, $\gamma$ is the skew of the clock, and $\kappa$ controls the failure probability. The above properties guarantee that the clocks (1)~have a skew bounded by $\gamma$ for polynomially many steps, w.h.p.; and (2)~in each step, the clocks make progress (at some node). A clock protocol $\vec C$ \emph{fails} at time $t$ if $\Delta(t) \ge \gamma$ occurs.
Several types of phase clocks have been proposed in the population protocol literature, e.g.~\cite{angluin2008fast-computation, gasieniec2018fast, alistarh2018space-optimal,sudo2019logarithmic,gasieniec2020time}. 

\paragraph{Bounded phase clocks via graphical load balancing.}
Let $G$ be a graph and suppose that each node of $G$ contains a bin, which is initially empty.
Consider the process, where in each step, a directed edge $(u,v)$ is sampled uniformly at random and a ball is placed into the \emph{least} loaded of bin among the two nodes connected by the edge (in case of ties, place the ball into bin $u$). Let $\ell(u,t)$ be the number of balls placed into bin $u \in V$ by the end of step~$t$ and use
\[
\Delta^*(t) = \max_{v \in V} \ell(v,t) - \min_{u \in V} \ell(u,t),
\]
to denote the gap between the most and least loaded bin. 
In~\appendixref{apx:process-gap}, we obtain the following bounds for this process on regular graphs, by leveraging 
the analysis of Peres et al.~\cite{peres2014graphical} for the above load balancing process. 

\begin{lemma}\label{lemma:unbounded-process-gap}
Let $G = (V,E)$ be a $d$-regular graph with $n$ nodes and edge expansion $\beta > 0$.
  For any constant $\kappa > 0$, there exists a constant $c(\kappa)$ such that for all $t > 0$ the gap satisfies
  \[
\Pr\left[ \Delta^*(t) > c(\kappa) \frac{d}{\beta} \log n  \right] < t/n^\kappa.
\]
\end{lemma}

We use the above result to obtain \emph{bounded} phase clocks in the $\ppmodel(G)$ model. 
We note that this is the only place in our framework where the initiator/responder distinction is used.
The proof of this result can be found in Appendix~\ref{app:proof-theorem-clocks}.

\begin{restatable}{theorem}{clockthm}\label{thm:clocks}
  Suppose $G = (V,E)$ is a $d$-regular graph with $n$ nodes and edge expansion $\beta > 0$.
  Let $\kappa > 1$ be a constant. Then for any $\gamma$ and $\phi$ satisfying
  \[
  \gamma \ge c(\kappa) \frac{d}{\beta} \log n \quad \textrm{ and } \quad \phi \ge 2\gamma
  \]
  there exists $(\phi,\gamma,\kappa)$-clock for $\ppmodel(G)$ that uses $\phi$ states per node.
\end{restatable}

\section{Simulating synchronous token shuffling protocols}
\label{sec:simulation}

In this section, we give our main technical result: synchronous protocols in the fully-connected token shuffling model can be simulated in the graphical, stochastic population protocol model.

\begin{restatable}{theorem}{simthm}\label{thm:simulation}
  Let $k > 0$ be a constant and $\vec A$ be a synchronous $k$-token shuffling protocol on $n$ nodes, where  $X$ is the set of local states and $Y$ the set of token types used the protocol $\vec A$. If $\vec A$ solves the task $\Pi$ with high probability in $R \in \poly(n)$ rounds, then there exists a stochastic population protocol~$\vec B$ that also solves task $\Pi$ with high probability on any $n$-node $d$-regular graph $G$ with edge expansion $\beta > 0$. The step complexity $T(\vec B)$ and state complexity $S(\vec B)$ of the protocol $\vec B$ satisfy
    \[
    T(\vec B) \in  O\left( R\cdot n \cdot \zeta \right) \quad \textrm{ and } \quad S(\vec B) \in O\left( |X| \cdot |Y|^k \cdot \zeta \right) \quad \textrm{ with } \quad \zeta = \log n \cdot \left( \frac{d}{\beta} + \frac{\taumix}{n} \right),
  \]
  where $\taumix$ is the mixing time of the $k$-stack interchange process on $G$.
\end{restatable}

\paragraph{Notation.}
The rest of this section is dedicated to proving this theorem. Throughout, we fix $R = R(n) \in \poly(n)$ and $\varepsilon = 1/n^a < 1/(R n^\lambda)$ for an arbitrary large constant $a > 0$. Let $G = (V,E)$ be $d$-regular $n$-node graph and $N = kn$. We use $\mu$ to denote the increment distribution of the $k$-stack interchange process on the graph $G$. The support of $\mu$ is the set $H \subseteq S_N$ and $\tau = \tau(\varepsilon)$ is the $\varepsilon$-mixing time of the $k$-stack interchange~process.

\subsection{The token shuffling protocol}

We now give a stochastic population protocol that simulates uniform schedules of the synchronous token shuffling model. The protocol simulates the random walk made by the $k$-stack interchange process, synchronised by phase clocks.

\paragraph{Setting up the clock.}
We choose the parameter $\kappa > 0$ such that a $(\phi,\gamma,\kappa)$-clock $\vec C$ with parameters given by
\[
\gamma \in \Theta\left( \frac{d}{\beta} \log n \right) \qquad \phi = \gamma + \vartheta \qquad \vartheta = \frac{2 \tau}{n} + 3 \gamma \qquad t^* = (R\phi + \gamma)n
\]
fails (i.e., the clock skew becomes $\gamma$ or greater) with probability at most $1/n^\lambda$ during the first $t^*$ steps.
Since $\phi \ge 2\gamma$, $R \in \poly(n)$, and $t^* \in \poly(n)$ hold, such a protocol exists by \theoremref{thm:clocks} for any constant $\lambda > 0$ by choosing a sufficiently large $\kappa$. The fact that $t^*$ is polynomially bounded follows from \theoremref{thm:d-beta-mixing} and that $\beta \ge 1/n^2$ for any regular connected graph. Further, 
$\tau \le \lceil \log 1/\varepsilon \rceil \cdot \taumix \in \poly(n)$, and hence, $\phi, \gamma \in \poly(n)$.

\paragraph{The token shuffling protocol.}
The parameter $\vartheta$ is used as a special threshold value for the token shuffling protocol.
We assume that each node $v$ holds exactly $k$ tokens, which are ordered from $0$ to $k-1$, in the same manner as cards ordered are in the $k$-stack interchange process. We say that the first token is the \emph{top token}.
We say that node $u$ is \emph{receptive} when ever its clock satisfies $c(u) < \vartheta$ and that it is \emph{suspended} otherwise. 
When nodes in $\{u,v\}$ interact, they apply the following rule:
\begin{enumerate}[noitemsep]
\item If both are receptive, that is, $c(u) < \vartheta$ and $c(v) < \vartheta$ holds, then
\begin{enumerate}[label=(\alph*)]
\item Let $q(u)$ and $q(v)$ be the random coin flips of $u$ and $v$, respectively.
\item If $q(u) = q(v) = 0$, then $u$ and $v$ swap their top tokens.
\item If $q(u) < q(v)$, then $v$ moves its top token to the bottom of its stack; $u$ does nothing.
\item If $q(u) = q(v) = 1$, then do nothing.
\end{enumerate}
\item Otherwise, do nothing.
\end{enumerate}
The protocol uses at most one random bit per node per interaction and that this is the only part of our framework, where the random bits provided to the nodes are used. The interacting nodes exchange at most 4 bits (i.e., whether they receptive or not, and the result of their coin flip) in addition to the contents of the swapped tokens in Step (1b). Finally, observe that when all nodes are receptive, the tokens are shuffled according to the increment distribution $\mu$ of the $k$-stack interchange process on $G$. \figureref{fig:shuffling} illustrates the dynamics of the shuffling protocol in the case~$k=1$.

\begin{figure}[t]
\centering
  \includegraphics[page=5,width=0.7\textwidth]{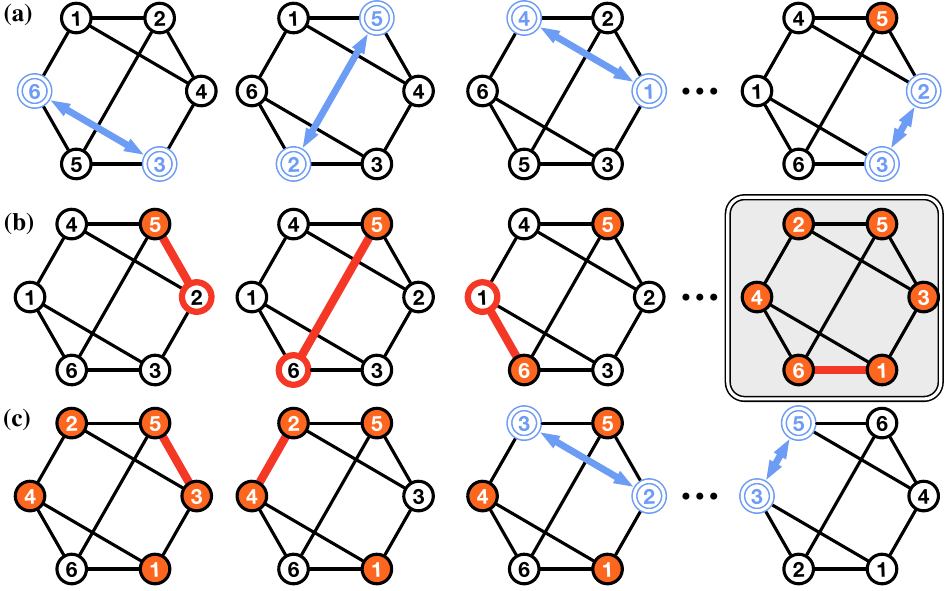}
  \caption{The dynamics of the shuffling protocol for $k=1$. Circles filled with white and red denote receptive and suspended nodes, respectively. The blue arrows connect nodes who exchange their tokens in the given step. Red lines denote steps, where at least one of the interacting nodes is suspended, and thus, no swap is made. 
  (a)~Initially all nodes are receptive and swap tokens with their interaction partners. After sufficiently many interactions, nodes become suspended and refrain from swapping tokens.
  (b)~Eventually all nodes are suspended. The highlighted panel shows the resulting permutation, which will act as the interaction pattern for the simulated round. (c)~As the phase clocks reset back to 0, nodes become receptive again, and the tokens are shuffled once more. \label{fig:shuffling}}
\end{figure}

\subsection{Analysis of the shuffling protocol}

We now analyse the above shuffling protocol.
Let $c(u,t)$ indicate the clock value of node $u$ at the end of step $t$. Let $c(u,0) = 0$ and $t(v,0) = 0$. We say that the clock of node $u$ resets at time step $t$ if its value transitions from $\phi-1$ to $0$. For $r \ge 0$, define
\begin{itemize}[noitemsep]
\item $t(v,r+1) = \min \{ t > t(v,r) : c(v,t) = 0 \}$; the step when $v$ resets its clock for the $r$th~time,
  \item $t_{\min}(r) = \min \{ t(v,r) : v \in V \}$; the \emph{earliest} step when some clock is reset for the $r$th time,
  \item $t_{\max}(r) = \max \{ t(v,r) : v \in V \}$; the \emph{latest} step when some clock is reset for the $r$th time.
\end{itemize}
Similarly, we define the times with respective to the events when the clocks reach the value~$\vartheta$:
\begin{itemize}[noitemsep]
\item $s(v,r) = \min \{ t > t(v,r) : c(v,t) = \vartheta \}$,
\item $s_{\min}(r) = \min \{ s(v,r) : v \in V \}$,
\item $s_{\max}(r) = \max \{ s(v,r) : v \in V \}$.
\end{itemize}
The following lemma captures the relationship between the timing of these events.

\begin{lemma}\label{lemma:timing}
  With high probability, the following inequalities hold:
  \begin{enumerate}[noitemsep]
  \item $t_{\max}(R+1) \le t^* = (R\phi + \gamma)n$,
  \item $s_{\min}(r) - t_{\max}(r) \ge \tau$ for each $1 \le r \le R$.
  \item $t_{\max}(r) < s_{\max}(r) < t_{\min}(r+1)$ for each $1 \le r \le R$.
  \end{enumerate}
\end{lemma}
\begin{proof}
  Recall that the clock protocol works correctly with high probability for the first $t^*$ steps. We now assume that this event occurs.

  For the first claim, we show that all nodes have incremented their clock at least $R\phi$ times after $t^*$ steps. For the sake of contradiction, suppose that some node $v$ has incremented its clock less than $R\phi$ times during the first $t^*$ steps. By the second property of the clock protocol, in every step $1 \le t \le t^*$, some node increments its clock value by one (modulo $\phi$). Hence, the nodes in $V \setminus \{ v \}$ have incremented their clocks at least $(R\phi + \gamma)(n-1)$ times. By the pigeonhole principle, some node $u \neq v$ has incremented its clock at least $R\phi + \gamma$ times. However, this contradicts the property that the difference in the clock skew is less than $\gamma$ for each step $1 \le t \le t^*$. Since each node has incremented its clock at least $R\phi$ times, each node has reset its clock $R$ times, so $t_{\max}(R+1) \le t^*$. 

For the second claim, observe that during an interval of $2 \gamma n$ steps there must exist a node that has incremented its clock $2\gamma$ times by the pigeonhole principle. By the first property of the clock protocol the skew is less than $\gamma$, so we get that
$t_{\max}(r) < t_{\min}(r) + 2 \gamma n$.
Again since the skew of the clock is less than $\gamma$, and in each step at most one node increments its clock counter, the time until some node reaches the clock value $\vartheta$ after step $t_{\min}(r)$ satisfies
$s_{\min}(r) \ge t_{\min}(r) + (\vartheta-\gamma)n$.
 Combining these two bounds and recalling that $\vartheta = \tau/n + 3 \gamma$ yields
 $s_{\min}(r) - t_{\max}(r) > (\vartheta-\gamma)n - 2 \gamma n = \tau$.
 Finally, the third claim follows from the fact that $\phi \ge 3 \gamma$ and that the skew is bounded by~$\gamma$.
\end{proof}

\paragraph{Distribution of tokens.} We now show that the distribution tokens mix to an $\varepsilon$-uniform distribution during the intervals $\{ t_{\max}(r)+1, \ldots, s_{\min}(r) \}$ for $1 \le r \le R$. Let $\pi_0 = \id$ and $\pi_t$ denote the locations of the tokens after $t$ steps of the shuffling protocol. Define $\sigma_0 = \id$ and
\[
\sigma_r = \pi_{s_{\max}(r)} \textrm{ for } 1 \le r \le R.
\]
Observe that $\sigma_{r} = \rho_3 \cdot \rho_2 \cdot \rho_1 \cdot \sigma_{r-1}$, where each $\rho_i$ is product of elements from the support $H \subseteq S_N$ of the increment distribution $\mu$ of the $k$-stack interchange process, where
\begin{itemize}[noitemsep]
\item $\rho_1 = x_{t_{\max}(r)} \cdots x_{t_{\min}(r-1)+1}$ (a subset of nodes have become receptive for the $r$th time),
\item $\rho_2 = x_{s_{\min}(r)} \cdots x_{t_{\max}(r)+1}$ (all nodes are receptive),
\item $\rho_3 = x_{s_{\max}(r)} \cdots x_{s_{\min}(r)+1}$ (a subset of nodes have become suspended for the $r$th time).
\end{itemize}
(Recall that permutations are applied from right to left.) Observe that while each $x_i$ is a random element of $H$, only the elements $\rho_2$ are guaranteed to be distributed according to the increment distribution $\mu$ of the $k$-stack interchange process. The elements of $\rho_1$ and $\rho_3$ are skewed towards the identity permutation, as some nodes are suspended whenever their clock values are in $\{ \vartheta, \ldots, \phi - 1 \}$. The next lemma establishes that this does not interfere with the mixing behaviour.

\begin{lemma}\label{lemma:epsilon-uniformity}
  Let  $0 \le r < R$. For any $A \subseteq S_N$, we have $\left| \Pr[ \sigma_{r+1} \in A \mid \sigma_r ] - \nu(A) \right| \le \varepsilon$.
\end{lemma}
\begin{proof}
  Suppose $\sigma_r$ is given.
  For brevity, let $\pi = \rho_2 \cdot \rho_1 \cdot \sigma_r$ so that $\sigma_{r+1} = \rho_3 \cdot \pi$. Define
\[
p(x) = \Pr[ \pi = x \mid \sigma_r ] \quad \textrm{ and } \quad p'(x) = \Pr[\sigma_{r+1} = x \mid \sigma_r].
\]
Observe that $|p(A) - \nu(A)| \le \varepsilon$, as $\rho_2$ is given by a sequence of at least $\tau$ elements sampled according to the increment distribution $\mu$. We show that $|p'(A)-\nu(A)| \le \varepsilon$. Let $y \cdot A$ denote the set $\{ yx : x \in A \}$. By expanding $p'(A)$ using conditional probabilities, we can write
\begin{align*}
  p'(A)  = \Pr[ \rho_3 \cdot \pi \in A \mid \sigma_r ] &= \sum_{y \in S_N} \Pr[  \pi \in y^{-1} \cdot A \textrm{ and } \rho_3 = y  \mid \sigma_r ] \\
  &= \sum_{y \in S_N} \Pr[ \rho_3 = y \mid \pi \in y^{-1} \cdot A \textrm{ and } \sigma_r] \cdot \Pr[ \pi \in y^{-1} \cdot A \mid \sigma_r] \\
  &= \sum_{y \in S_N} q(y) \cdot p(y^{-1} \cdot A),
\end{align*}
where $q(y) = \Pr[ \rho_3 = y \mid \pi = y^{-1} x, \sigma_r]$ is a probability distribution on $S_N$. Hence, $q(S_N) = \sum q(y) = 1$.
Since $\nu(A) = \nu(z\cdot A)$ for any $z \in S_N$, it follows that
\begin{align*}
  \left| p'(A) - \nu(A) \right| &= \left| \sum_{y \in S_N} q(y) \cdot p(y^{-1} \cdot A) - \nu(A) \right| = \left| \sum_{y \in S_N} q(y) \cdot \left[ p(y^{-1} \cdot A) - \nu(y^{-1} \cdot A) \right] \right| \\
  &\le \sum_{y \in S_N} q(y) \left| p(y^{-1} \cdot x) - \nu(y^{-1} \cdot A)  \right| 
  \le \sum_{y \in S_N} q(y) \cdot \varepsilon \le \varepsilon.\qedhere
\end{align*}
\end{proof}

\subsection{The simulation protocol}

Using the shuffling protocol in the population protocol model, we can simulate an $R$-round algorithm~$\vec{A}$ in the synchronous $k$-token shuffling model. Let $f \colon X \times Y^k \to X$ be the state transition function and $g \colon X \to Y^k$ be the token generation function of the algorithm $\vec A$. Recall that $X$ and $Y$ denote the sets of local states and token types, respectively.

\paragraph{The simulation protocol.}
Each node $v$ maintains the following variables:
\begin{itemize}[noitemsep]
\item $a(v) \in X$ to simulate the local state of the synchronous protocol $\vec A$,
\item $b_0(v), \ldots, b_{k-1}(v) \in Y$ to store the sent and received tokens, and
\item $r(v) \in \{0, 1, \ldots, R\}$ to store the number of simulated rounds.
\end{itemize}
The variable $a(v)$ is initialised to the initial state $x_0(v)$ of node $v$ in the algorithm $\vec A$ and $b_0(v), \ldots, b_{k-1}(v)$ are initialised to the values given by $g(x_0(v))$. The variable $r(v)$ is initially set to 0. When node $v$ interacts (in the asynchronous population protocol model), $v$ updates its state according to the following rules:
\begin{enumerate}
\item Run the clock and the shuffling protocol using $b_0(v), \ldots, b_{k-1}(v)$ to hold the $k$ tokens.
\item If $c(v) = \vartheta$, then 
  \begin{itemize}
  \item update the round counter and set $r(v) \gets \max \{ r(v)+1, R \}$,
  \item compute the new state $a(v) \gets f\left( a(v), b_0(v), \ldots, b_{k-1}(v) \right)$, and
  \item generate new tokens $b_0(v), \ldots, b_{k-1}(v) \gets g(a(v))$.
  \end{itemize}
\end{enumerate}
As output value of the simulation, node $v$ uses the output value algorithm $\vec A$ associates to state $a(v)$. 
The above algorithm simulates an execution of the synchronous algorithm $\vec A$ under the schedule  $\sigma_1, \ldots, \sigma_R$ given by the shuffling protocol. To this end, define $x_0(v) = a(v,0)$ and $x(r) = a(v, s(v,r))$ for  all $1 \le r \le R$. The proof of the next lemma is given in \appendixref{apx:executions}.

\begin{lemma}\label{lemma:executions}
With high probability, the sequence $( x_r )_{0 \le r \le R}$ is an execution induced by the schedule $( \sigma_r )_{1 \le r \le R}$.
\end{lemma}

\subsection{From almost-uniform schedules to uniform schedules}

The schedules provided by the shuffling protocol are only $\varepsilon$-uniform, as the shuffling process is executed for finitely many steps. We now show that this does not matter: any synchronous protocol behaves statistically similarly under $\varepsilon$-uniform and uniform schedules.

To formalise this, let $\Phi$ be the distribution over sequences $(\sigma_1, \ldots, \sigma_R) \in S^R_N$ of permutations generated by the shuffling protocol under the assumption that the clock protocol works correctly for $T$ time steps. Let $\nu^R = \nu \times \cdots \times \nu$ denote the distribution of a sequence of $R$ independently and uniformly sampled random permutations from $S_N$. That is, $\nu^R$ is the distribution of the uniform $R$-round schedules. The following then holds: 

\begin{lemma}\label{lemma:prefix-statistical-distance}
  The total variation distance between $\Phi$ and $\nu^R$ satisfies $\tvnorm{ \Phi - \nu^R } \le \varepsilon R$.
\end{lemma}
\begin{proof}
  Let $A = A_1 \times \cdots \times A_r \subseteq S^R_N$.
Since the sequence $\sigma_1, \ldots, \sigma_R$ is Markovian, we have
  \[
\Phi(A) = \Pr\left[(\sigma_1, \ldots, \sigma_R) \in A \right] = \Pr\left[\sigma_1 \in A_1 \right] \cdot \prod_{j=2}^R \Pr\left[\sigma_i \in A_i \mid \sigma_{i-1} \in A_{i-1}\right] = \prod_{j=1}^R \phi_j(A),
\]
where $\phi_i(A) = \Pr[\sigma_i \in A_i \mid \sigma_{i-1} \in A_{i-1}]$ for $i > 0$. Recall that $\sigma_0 = \id$.
For notational convenience, let $\nu_i(A) = \nu(A_i)$. 
Next, we make use of the following inequality (see \appendixref{apx:prod-upper-bound} for a proof). For any $a_i, b_i \in \mathbb{R}^+$, where $1 \le i \le t$, we have that
  \[
  \left | \prod_{i=1}^t a_i - \prod_{j=1}^t b_j \right | \le \sum_{i=1}^t |a_i - b_i| \left( \prod_{k=1}^{i-1} a_k \right) \left( \prod_{h=i+1}^{t} b_h \right).
  \]
By applying the above identity, we obtain
\begin{align*}
  \left|\Phi(A) - \nu^R(A) \right| &\le \sum_{i=1}^R \left| \phi_i(A) - \nu_i(A) \right| \left( \prod_{k=1}^{i-1} \phi_k(A) \right) \left( \prod_{h=i+1}^{R} \nu_h(A) \right) \\
  &\le  \sum_{i=1}^R \left| \Pr[\sigma_i \in A_i \mid \sigma_{i-1} \in A_{i-1}] - \frac{|A_i|}{N!} \right| \left( \prod_{k=1}^{i-1} \phi_k(A) \right) \left( \prod_{h=i+1}^{R} \nu_h(A) \right) \\
  &\le   \sum_{i=1}^R \varepsilon \left( \prod_{k=1}^{i-1} \phi_k(A) \right) \left( \prod_{h=i+1}^{R} \nu_h(A) \right) 
  \le \sum_{i=1}^R \varepsilon \le \varepsilon R,
\end{align*}
where the third inequality follows from \lemmaref{lemma:epsilon-uniformity} and the second last from the fact that the products are over probabilities.
The claim now follows as
\[
 \tvnorm{ \Phi - \nu^R } = \max_{A \subseteq S^R_N} |\Phi(A) - \nu^R(A)| \le \varepsilon R. \qedhere
\]
\end{proof}

\subsection{Proof of the simulation theorem}

We are now ready to show our main technical result. First, recall the following property.

\begin{lemma}\label{lemma:tv-max-fun}
  Let $\mu$ and $\nu$ be probability distributions over a finite domain $\Omega$.
  For any function $F \colon \Omega \to \Omega'$, the total variation distance satisfies $\tvnorm{ F(\mu) - F(\nu) } \le \tvnorm{ \mu - \nu }$.
\end{lemma}

With all the pieces now in place, we can now state and prove our simulation theorem.

\simthm*
\begin{proof}
  Let $\vec A$ be the synchronous $k$-token shuffling protocol. Since the protocol works with high probability, assume it succeeds with probability at least $p \ge 1- 1/n^h$, where $h$ is a constant we choose later. Using the simulation protocol, we construct a graphical population protocol $\vec B$ with the claimed properties. Recall that $\varepsilon = 1/n^a$, where $a$ was an arbitrary constant. We set $a$ so that $\varepsilon R \le 1/n^\lambda$ holds.
  By \lemmaref{lemma:executions} the shuffling protocol simulates the execution of $\vec A$ induced by an $R$-round $\varepsilon$-uniform synchronous schedule $( \sigma_r )_{1 \le r \le R}$ with high probability. This takes at most $t^* = (R\phi + \gamma)n$ steps by \lemmaref{lemma:timing}.

  Recall that $\phi \in O(\gamma + \tau/n)$, where $\gamma$ is the bound on the clock skew and $\tau = \tau(\varepsilon)$ is the $\varepsilon$-mixing time of the $k$-stack interchange process. Since $\tau \le \lceil \log 1/\varepsilon \rceil \cdot \taumix \in O\left( \log n \cdot \taumix \right)$, we get from 
  \theoremref{thm:clocks} the following bounds:
  \[
  \gamma \in O\left( \frac{d\log n}{\beta} \right) \qquad
  \phi \in O\left( \log n \left( \frac{d}{\beta} + \frac{\taumix}{n} \right)\right) \qquad
  t^* \in O\left( R n \log n \left( \frac{d}{\beta} +  \frac{\taumix}{n} \right) \right).
  \]
  The bound on $t^*$ establishes the claimed bound on the step complexity of $\vec B$. For the state complexity, note that each node $v$ stores the variables for the clock $c(v) \in [\phi]$, the round counter $r(v) \in [R+1]$, and the local state $a(v) \in X $ of the simulated protocol $\vec A$, and the $k$ tokens $b_0(v), \ldots, b_{k-1}(v) \in Y$.
  This takes $\phi \cdot (R+1) \cdot |X| \cdot |Y|^k$ states, establishing the bound on the state complexity.

  It remains to argue that $\vec B$ solves the task $\Pi$ with probability at least $p-1/n^\lambda$. The output of algorithm $\vec A$ on input $z$ under any $R$-round synchronous schedule $\Xi$ is given by $F_z(\Xi)$, where $F_z$ is a computable function. Let  $D = F_z(\Phi)$ and $D' = F_z(\nu^R)$ be the probability distributions of outputs in the executions of the algorithm $\vec A$ induced, respectively, by the simulated $\varepsilon$-uniform schedules given by $\Phi$ and the uniform schedules given by $\nu^R$. By \lemmaref{lemma:prefix-statistical-distance} and \lemmaref{lemma:tv-max-fun}, we have that
\[
\tvnorm{D - D'} = \tvnorm{F_z(\Phi) - F_z(\nu^R)} \le \tvnorm{\Phi - \nu^r} \le \varepsilon R \le 1/n^\lambda.
\]
Therefore, the probability that the output $z'$ under the execution induced by the $\varepsilon$-uniform schedule on input $z$ satisfies $z' \in \Pi(z)$ is
$D(\Pi(z)) \ge D'(\Pi(z)) - 1/n^\lambda \ge p - 1/n^\lambda \ge 1 - 1/n^h - 1/n^\lambda$.
Thus, the output of protocol $\vec B$ is feasible with high probability. 
Since $\vec A$ stabilises in $R$ rounds, nodes can set the output of $\vec B$ to be the output of $\vec A$ at the end of the $R$th simulated round. Thus, the output of $\vec B$ stabilises as well.
\end{proof}

\section{Applications: leader election and exact majority}
\label{sec:applications}

Using~\theoremref{thm:simulation}, we can automatically transport algorithms from
the \emph{fully-connected synchronous token shuffling model} to the graphical, asynchronous population protocol model. We now utilise this result to obtain fast protocols for leader election and exact majority in the graphical population protocol model. To this end, we give the following algorithms in the token shuffling model:
\begin{enumerate}
\item A leader election algorithm that uses one-way communication with $k > 1$ tokens. The protocol uses a one-way information dissemination protocol and a protocol for generating synthetic coins in the token shuffling model.

\item An exact majority algorithm simulating two-way interactions in a population of $2n$ virtual agents. The algorithm uses the classic cancellation-doubling dynamics.
\end{enumerate}
These protocols adapt ideas from prior work in the clique model (see e.g.~\cite{elsaesser-survey} for a general overview). 
For completeness, \appendixref{apx:token} provides the full analyses of the algorithms in the synchronous token shuffling model.

\subsection{Warmup: one-way information dissemination}

We start by adapting a classic broadcast primitive to the $k$-token shuffling model. 
We assume that each node is given an input $z(v)$ from a set $\Sigma$ with a total order on the values. The protocol computes the maximum value given as input. This can be used for information dissemination, or to agree on a common input value.

\paragraph{One-way epidemics protocol.}
The algorithm works in the $k$-token shuffling model for any $k > 0$. At the start of the protocol, each node $v$ initialises a local state variable $a(v)$ to its input value $z(v)$. In every round, each node $v$ performs the following steps:
\begin{enumerate}[noitemsep]
\item Generate $k$ tokens of type $a(v)$.
\item Use one round to shuffle the generated tokens.
\item After receiving $k$ tokens $y_0, \ldots, y_{k-1}$, set $a(v) \gets \max \{ a(v), y_0, \ldots, y_{k-1} \}$.
\end{enumerate}
The number of states and token types used by the algorithm is $|\Sigma|$.

\begin{lemma}\label{lemma:consensus}
After $O(\log n)$ rounds, every node $v$ satisfies $a(v) = \max z(u)$ w.h.p.
\end{lemma}

\subsection{Leader election by token-shuffling}
\label{ssec:le}

We now consider the leader election problem, where the goal is to select a single node as a leader. We adapt a well-known strategy from the standard population protocol model: each leader candidate iteratively (1) flips a random coin and (2) becomes a follower if another leader candidate had a coin flip with a larger value~\cite{gasieniec2018fast,elsaesser-survey}. 

In order to implement step (1) we need access to random bits. However recall that by our definitions, the state transition and token generation functions in the token shuffling model are deterministic. While we could ``lift'' random bits from the underlying stochastic population protocol model, we instead opt for generating synthetic coin flips in the $k$-token shuffling model with $k > 1$.

\paragraph{Synthetic coin flips.}
Let $k > 1$ and consider the $k$-token shuffling model, where each node $v$ receives exactly $k$ tokens. Recall that these tokens are ordered from $0$ to $k-1$. We leverage this property to generate synthetic coin flips in one round as follows:
\begin{enumerate}[noitemsep]
\item Each node generates a single 0-token and a single 1-token.
\item Use one round to shuffle the generated tokens.
\item Output the value of the first token.
\end{enumerate}
While the coin flips between nodes are not independent, the probability that a node outputs $1$ is $1/2$. Thus, in expectation half of the nodes output 1.

\paragraph{Leader election protocol.}
Suppose 
the input specifies a (nonempty) subset of nodes that start as leader candidates. Let $\ell(v) \in \{0,1\}$ be a local variable of node $v$ denoting whether it considers itself a leader candidate. In a single iteration, each node $v$ executes the following:
\begin{enumerate}[noitemsep]
\item Generate a synthetic coin flip $b(v)$ in one round.
\item Run $\Theta(\log n)$ rounds of the broadcast protocol of \lemmaref{lemma:consensus} with input $\ell(v) \cdot b(v) \in \{0,1\}$.
\item If $\ell(v) = 1$ and $b(v) = 0$, then set $\ell(v) \gets 0$ if the broadcast protocol had output 1.
\end{enumerate}
In every round, node $v$ uses the value $\ell(v)$ as its current output value. Each iteration of the protocol takes $\Theta(\log n)$ rounds and the protocol uses $\Theta(\log n)$ states and constantly many token types. We show that with high probability, the protocol reduces the number of leader candidates to one after $O(\log n)$ iterations, and hence, in $O(\log^2n)$ rounds. The remaining candidate is the elected leader.

Note that any node that is a leader candidate ceases to be a leader candidate only if its local coin flip was 0 \emph{and} the broadcast protocol informs the node that some other leader candidate had a local coin flip with value one. Thus, we never end up in a situation where there are no leader candidates remaining.

\begin{restatable}{theorem}{LEalg}\label{thm:le}
There is a synchronous 2-token shuffling protocol for the leader election task that stabilises in $O(\log^2 n)$ rounds with high probability, uses $O(\log n)$ states per node and two token types.
\end{restatable}

By applying \theoremref{thm:d-beta-mixing} and \theoremref{thm:simulation}, we get the following result.

\begin{corollary}
There exists a stochastic population protocol that solves leader election in $(d/\beta)^2 \cdot n \polylog n$ steps with high probability using $(d/\beta)^2  \cdot \polylog n$ states in any $d$-regular graph with edge expansion $\beta > 0$.
\end{corollary}

\subsection{Exact majority by token-shuffling}
\label{ssec:majority}

We now show a protocol for the exact majority task in the 2-token shuffling model.
Specifically, we simulate a cancellation-doubling protocol population protocol in a population of size $2n$, where nodes interact synchronously according to a randomly chosen perfect matching. In every round, each node receives two tokens of type $A$ and $B$, and generates two new tokens for the next round by applying a rule of the form
\[
 A + B \to C + D.
 \]
 The rules used guarantee that with high probability all tokens get converted to the value held by the initial majority of input values. Hence, as output of the protocol, each node $v$ can use (an arbitrary) value held by one of its tokens.

\paragraph{The exact majority protocol.} Let $N=2n$ and $t = (\lambda+1) \log_{5/4} N$, where $\lambda > 0$ is an arbitrary constant. Each node $v$ initially creates two tokens that take the input value $z(v) \in \{0,1\}$. After this, the algorithm consists of repeatedly running the following rules:
\begin{enumerate}[noitemsep]
\item For $t$ consecutive rounds, apply the cancellation rules
  \[
  Z + \bar{Z} \to \emptyset + \emptyset \textrm{ for } Z \in \{0,1\}.
  \]

\item For $t$ consecutive rounds, apply the doubling rules
  \[
  Z + \emptyset \to Z^{1/2} + Z^{1/2} \textrm{ for } Z \in \{0,1\}.
  \]

\item Apply the promotion rule $Z^{1/2} \to Z$ for each token of type $Z \in \{0,1\}$.

\end{enumerate}
Step (1) is called the \emph{cancellation phase} and Step (2) the \emph{doubling phase}. The protocol uses exactly five types of tokens: $0, 1, 0^{1/2}, 1^{1/2}$ and $\emptyset$. Tokens of type $Z \in \{0,1\}$ represent an ``opinion'' on what is the majority value. The tokens of type $Z^{1/2}$ are called \emph{split} tokens. A token of type $\emptyset$ is called an \emph{empty token}. The idea is that (1) opposing opinions cancel out during the cancelling phase and (2) the amount of majority tokens doubles in the doubling phase. 

\begin{figure}[t]
  \includegraphics[page=1,width=\textwidth]{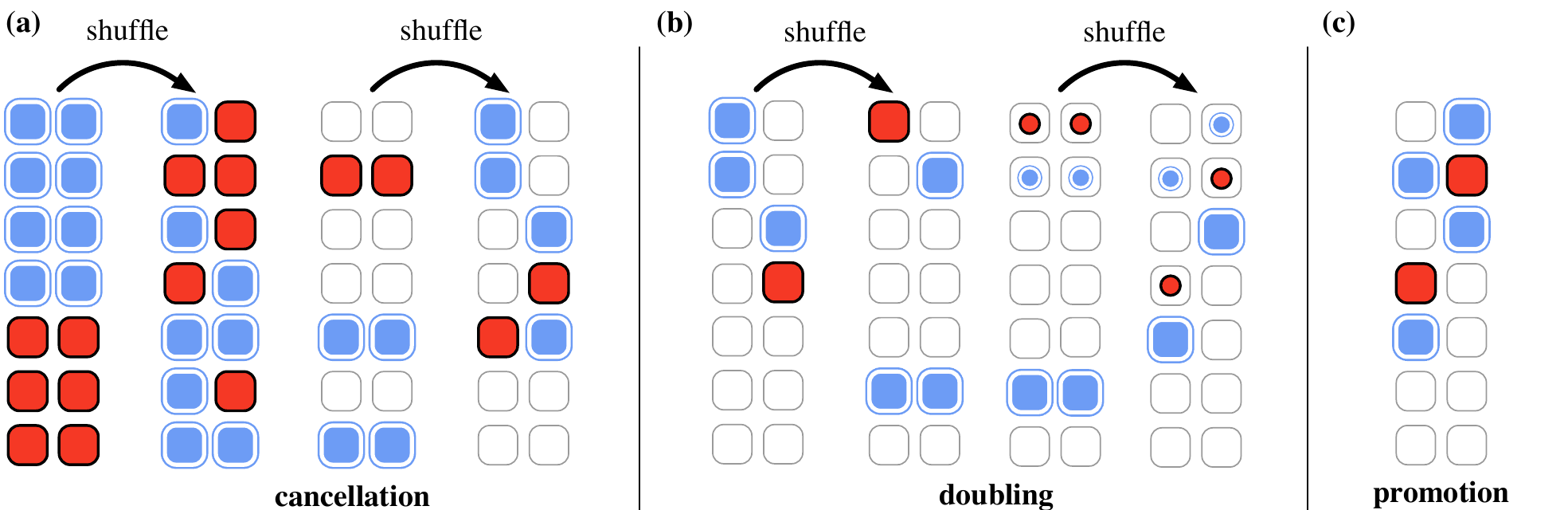}
  \caption{The cancellation-doubling dynamics with $2n$ tokens and $n=7$ nodes. Blue tokens have the initial majority. (a) A single round of the cancellation phase. White rectangles represent empty tokens. (b) Two rounds of the doubling phase. The small circular tokens are split tokens. (c) The promotion rule promotes all split tokens into full tokens at the end of the doubling phase.  \label{fig:cancellation-doubling}}
\end{figure}

In each round, every node holds two tokens $y_0$ and $y_1$. If one of the tokens is nonempty, then the node outputs the largest value held by nonempty tokens. Otherwise, if both tokens of a node are empty, i.e., $y_0 = y_1 = \emptyset$, then the node outputs its input value (in this case the protocol has not yet stabilised).
We show that the algorithm stabilises in $O(\log^2 N) = O(\log^2 n)$ rounds with high probability, i.e., the system reaches a configuration, where all generated tokens take the majority input value.

\begin{restatable}{theorem}{MAJalg}\label{thm:majority}
  There is a synchronous 2-token shuffling protocol for the exact majority task that stabilises in $O(\log^2 n)$ rounds with high probability, uses $O(\log n)$ states and five token types.
\end{restatable}

By applying \theoremref{thm:d-beta-mixing} and \theoremref{thm:simulation}, we get the following result.

\begin{corollary}
There exists a stochastic population protocol that solves exact majority in $(d/\beta)^2 \cdot n \polylog n$ steps with high probability using $(d/\beta)^2  \cdot \polylog n$ states in any $d$-regular graph with edge expansion $\beta > 0$.
\end{corollary}

\subsection{Backup protocols}

Finally, we address the following technical detail: our simulation framework and the simulated synchronous algorithms are guaranteed to work correctly and stabilise only with high probability, and therefore, the protocols may fail with low probability. To obtain always correct protocols, i.e., ones with finite expected stabilisation time, we specify ``backup protocols'', which are run in the unlikely cases, where either the simulation framework fails (e.g.\ the phase clocks become desynchronised) or the fast synchronous algorithm fails. 
This problem also occurs in the context of fast clique-based algorithms, e.g.~\cite{ alistarhg15,alistarh2018space-optimal}, and we adopt similar mitigation strategies. 

Note that since the probability of failure of the fast protocols can be polynomially small, to get polynomial expected stabilisation time, it suffices to have a backup protocol that has polynomial expected stabilisation time and small state complexity.

\paragraph{Backup for exact majority.}
The backup protocol, if necessary, is initiated as follows. If some node notices disagreement or inconsistent states after the fast protocols supposed stabilisation time, it initiates a signalling message, which is propagated further by all nodes that receive it. This signal forces all nodes to switch to executing the reliable (but potentially slow) backup protocol. Since the backup is only executed with low probability, and has negligible space cost, it does not affect the overall complexity of the fast exact majority protocol in the graphical population protocol model. 

In the case of the exact majority protocol, we can directly adopt the same solution as in the classic clique setting~\cite{alistarh2018space-optimal}: use the four-state exact majority algorithm analysed by Draief and Vojnovi\'c~\cite{draief2012-convergence} as a backup protocol. This algorithm works in arbitrary, connected graphs and has polynomial expected stabilisation time.

\paragraph{Backup for leader election.}
For leader election, we use the six-state leader election algorithm given by Beauquier et al.~\cite{beauquier2013self} who studied this protocol under the adversarial (non-stochastic) scheduler. In \sectionref{sec:constant-le}, we show that this protocol has polynomial expected stabilisation time under the stochastic scheduler on any connected graph. 

Switching to the backup protocol can be done as follows: once a node has executed the fast protocol sufficiently many rounds, it switches to the slow protocol using its current state (whether it is a leader candidate or not) as input for the constant-state protocol. Any node that observes during an interaction that some other node has switched to the slow protocol, does so as well.

\section{Convergence analysis for leader election on general graphs\label{sec:constant-le}}

In this section, we analyse the leader election protocol of Beauquier et al.~\cite{beauquier2013self} under the uniform stochastic scheduler. We establish the following result.

\begin{theorem}\label{thm:token-le}
There exists a protocol for leader election that uses six states and stabilises in any graph $G$ in $O(\diam(G) \cdot n^3 m\log n)$ interactions w.h.p. and in expectation, where $n$ is the number of nodes, $m$ the number of edges, and $\diam(G)$ the diameter of $G$.
\end{theorem}

\subsection{The token-based leader election protocol}

First, we recall that in the classic clique setting leader election can be solved by a simple 2-state protocol, where each node keeps track of whether it is a leader candidate or a follower. Whenever two leader candidates interact, the initiator stays as a leader candidate while the responder becomes a follower; no other type of interaction changes the state of nodes. 

The token-based protocol uses a similar approach. However, unlike in the clique, it may be impossible for two leader candidates to directly interact: they may not have a common edge in $G$. Instead, the nodes use tokens to interact indirectly. In each step, the nodes update their status by exchanging tokens between their interaction partners and at every time step each node holds exactly one token. 

There are three types of tokens: black, white, and inactive tokens. Initially, each leader candidate creates a black token. In each step, nodes exchange their tokens. Whenever two black tokens meet, exactly one of them turns into a white token while the other remains black. Informally, black tokens represent the presence of a leader candidate that has not been yet cancelled. A white token represents a leader candidate that will \emph{eventually} become a follower. Whenever a node that considers itself a leader candidate receives a white token, it changes its own status into a follower and deactivates the token. The invariant maintained by the protocol is that the total number of non-inactive tokens present in the system equals the number of leader candidates. By continuously shuffling the tokens, it is eventually guaranteed that the total number of black tokens becomes one and all other tokens become inactive.

\paragraph{The protocol.}
Formally, the state of each node $v$ is a tuple $(\ell, y)$, where $\ell \in \{\mathsf{leader},\mathsf{follower}\}$ is a bit indicating whether node $v$ is a leader candidate and $y \in \{ \mathsf{black}, \mathsf{white}, \mathsf{inactive}\}$ denotes the type of the token held by the node. As input, each node is given a bit indicating whether it is a leader candidate initially. Every node $v$ initialises its state using the following rules:
\begin{itemize}
    \item If $v$ is a leader candidate, then it sets $\ell(v) \gets \mathsf{leader}$ and $y(v) \gets \mathsf{black}$.
    \item Otherwise, it sets $\ell(v) \gets \mathsf{follower}$ and $y(v) \gets \mathsf{inactive}$.
\end{itemize}
When two neighbouring nodes $u$ and $v$ are selected to interact by the scheduler, we say that (also) the tokens held by the nodes interact. 
On every interaction, where node $u$ is the initiator and $v$ is the responder, the states are updated as follows:
\begin{enumerate}
    \item If $y(u)=y(v)=\mathsf{black}$ holds, then $y(v) \gets \mathsf{white}$. That is, if both tokens are black, then the token of the responder $v$ is coloured $\mathsf{white}$. 
    \item If the token held by $u$ is $\mathsf{white}$, $y(u) = \mathsf{white}$, and node $v$ is a $\mathsf{leader}$, $\ell(v) = \mathsf{leader}$, then 
    \begin{itemize}
        \item node $v$ designates itself as $\mathsf{follower}$, i.e., sets $\ell(v) \gets \mathsf{follower}$.
        \item node $u$ sets the type of its token to $y(u) \gets \mathsf{inactive}$.
    \end{itemize}
    \item Finally, the nodes $u$ and $v$ swap their tokens $y(u)$ and $y(v)$.
\end{enumerate}
For the formal proof of correctness, we refer to~\cite{beauquier2013self}. Here, we focus only on bounding the time for the protocol to stabilise under the uniform stochastic scheduler on $G$. 

\subsection{Bounding the hitting and meeting times of tokens}

To establish bounds on the stabilisation time of the leader election protocol, we analyse the hitting time and meeting time of tokens performing random walks on the graph $G$. Later, the stabilisation time of the above token-based leader election protocol can be bound using these quantities.
Before we proceed, we note the differences between the classic random walk process on a graph and the random walks made by the tokens in our process. 

\paragraph{Random walks on graphs.}
Recall that the classic random walk on a graph $G$ is the following Markov chain: Initially, a random walker (i.e.\ a token) is placed to some node $v$ of $G$. In each step, the random walker moves from $v$ to a some neighbor $u$ of $v$ chosen uniformly at random. A natural extension is to consider multiple, independent random walkers moving on the nodes of $G$: there may be several walkers  placed on nodes of $u$ and in every step each walker moves to a new random node independently of all the other nodes.

In contrast, in the population protocol model, we have to consider multiple tokens performing correlated random walks on $G$: in every step exactly two tokens move along the same edge, which is sampled uniformly at random. Nevertheless, we can carefully adapt and use analogous arguments to analyse the classic random walk (see e.g.~\cite{lovasz1993random}) and coalescence time of independent random walks as used by Cooper et al.~\cite{cooper2013coalescing}.
Naturally, the bounds we obtain are somewhat different, as the underlying sampling process is different, and we do not aim for sharp bounds.

\paragraph{Hitting times for irreducible Markov chains.}
We start by recalling the following elementary result about hitting times of Markov chains; see e.g.~\cite[Proposition 1.19]{levin2017markov}.
For states $x$ and $y$, the expected hitting time $H(x,y)$ between is 
\[
 H(x,y) = \E[ \min \{ t \ge 1 : X_t = y, X_0 = x\} ].
\]
For $x \neq y$, $H(x,y)$ is the expected number of steps until the chain starting in state $x$ reaches state $y$. For $x = y$ the value gives the expected first return time to state $x$.

\begin{lemma}\label{lemma:hitting}
For any finite and irreducible Markov chain, the stationary distribution $\pi$ satisfies
\[
 \pi(x) =\frac{1}{H(x,x)} \qquad \textrm{ for every state } x.
\]
\end{lemma}
Note that the above lemma does not require that the Markov chain is aperiodic. Indeed, the chains we consider will be periodic.

\paragraph{Random walk of a single token.}
Let $G = (V,E)$ be a simple, connected graph on $n \ge 2$ nodes, with $m \geq 1$ edges. 
We start by analysing the walk performed by a single token of the leader election algorithm under the population protocol model.
More precisely, we consider the following process. Initially, a token placed on a node of $G$. In each time step, an edge of $G$ is sampled uniformly at random. If the token is located at an endpoint of the sampled edge, then it moves to the other endpoint of that edge. Otherwise, the token stays put.

Formally, this corresponds to a Markov chain on the state space $V$. The probability that the chain transitions from state $u$ to $v$ is given by 
\[
P(u,v) = \begin{cases}
1/m & \textrm{if } \{u,v\} \in E \\
1 - d(u)/m & \textrm{if } u = v \\
0 & \textrm{otherwise}
\end{cases}
\]
for every $u,v \in V$. That is, if the token is on node $u$, the probability for the token to move to node $v$ in the next step is the probability that an edge between the two nodes, if existent, is chosen. 

The resulting Markov chain is irreducible, since the graph $G$ is connected. Note that this random walk differs from the classic random walk on $G$, where the token moves lazily to an adjacent node in each time step. First, we show that the uniform distribution on $V$ is the stationary distribution for this chain. 

\begin{lemma}\label{lemma:stationary-pp-walk}
The stationary distribution of the walk on $G$ is $\pi(v) = 1/n$ for every $v \in V$.
\end{lemma}
\begin{proof}
Let $\pi$ be the uniform distribution on $V$. Recall that $\pi P$ denotes the application of the transition matrix $P$ to $\pi$. To establish our claim, we need to validate that $\pi = \pi P$ holds. For this, we observe that 
\[
(\pi P)(v) = \sum_{u } \pi(u) P(u,v) 
               =  \frac{1}{n} \sum_{u} P(u,v) = \frac{1}{n} = \pi(v).\qedhere
\]
\end{proof}
Next, using elementary arguments, we can bound the hitting time of any pair of nodes; see e.g.~\cite{lovasz1993random}.
In particular, we make use of the worst-case expected hitting time defined by
\[
H_\textrm{max} = \max \{ H(u,v) : u,v \in V, u \neq v\}.
\]
The following lemma shows that $H_\textrm{max} < \diam(G) \cdot nm$.

\begin{lemma}
For any graph $G$, $H(u,v) < \diam(G) \cdot nm$ for all $u,v \in V$. \label{lemma:hitting-time}
\end{lemma}
\begin{proof}
By \lemmaref{lemma:hitting} and \lemmaref{lemma:stationary-pp-walk}, we have $H(u,u) = 1/\pi(u) = n$. On the other hand, by calculating the expected hitting time in another way, we observe that 
\[
H(u,u) = 1-\frac{d(u)}{m} + \frac{1}{m} \sum_{\{u,w\} \in E} \left( 1 + H(u,w) \right) = n.
\] 
Thus, we get the inequality
\begin{align*}
 \sum_{\{u,w\} \in E} \left( 1 + H(u,w) \right) < nm.
\end{align*}
In particular, we have that $H(u,w) < nm$ for any edge $\{u,w\} \in E$. Since $G$ has diameter $D$, there is a path of length $u = u_0, \ldots, u_k = v$ at most $\diam(G)$ between any two nodes $u$ and $v$. Hence, by linearity of expectation, we get that
\[
H(u,v) \le \sum_{i=0}^{k-1} H(u_i,u_{i+1}) < \diam(G) \cdot nm. \qedhere
\]
\end{proof}

\paragraph{Meeting time of two tokens.}
We now consider the situation where a distinct token is placed on each node of $G$.
In each time step, an edge is chosen uniformly at random. Whenever the edge $\{u,v\}$ is sampled, the tokens at nodes $u$ and $v$ exchange places. Note that, individually, each of the tokens performs a random walk, but the random walks are \emph{not} independent. 

We say that two tokens \emph{meet} at time $t$ if the edge $\{u,v\}$ is sampled at time step $t$ and the two tokens are located at the nodes $u$ and $v$, respectively. From now on, we uniquely label the tokens from $0$ to $n-1$ and define the random variable $M(a,b)$ as the number of time steps until tokens $a$ and $b$ first meet, starting in the initial configuration. If $a=b$, then we follow the convention that $M(a,b)=0$. We are interested in bounding the largest first meeting time between any two pairs of tokens. To this end, we define
\[
M = \max \{ M(a,b) : a,b \in [n] \} \quad \textrm{ and } \quad M_\textrm{max} = \max_{a,b} \left \{ \E[M(a,b)] \right \}.
\]
The random variable $M$ is the largest first meeting time between any pairs of tokens in the token shuffling process and the quantity $M_\textrm{max}$ is the worst-case expected first meeting time between any two tokens.

\begin{lemma}\label{lemma:mmax-bound}
The expected worst-case first meeting time satisfies $M_\textrm{max} \in O(\diam(G)\cdot n^3 m)$.
\end{lemma}
\begin{proof}
We keep track of the locations of the two tokens and the \emph{parity} of the number of times the two tokens have met. To this end, we define the graph $G^* = (V^*, E^*)$ with 
\[
V^* = V^*_0 \cup V^*_1, \textrm{ where } V^*_b = \{ ( \{v_0,v_1\} ,b ) : v_0,v_1 \in V, v_0 \neq v_1 \} \textrm{ for } b \in \{0,1\}
\]
and $\{ (\{u_0,u_1\},b), (\{v_0,v_1\},b' ) \} \in E^*$ if either of the following two conditions hold:
\begin{enumerate}[noitemsep]
    \item $b=b'$, $u_i = v_i$ and $\{v_{1-i}, u_{1-i}\} \in E$ for some $i \in \{0,1\}$, or
    \item $b \neq b'$ and $\{u_0,u_1\} = \{v_0,v_1\}$. 
\end{enumerate}
One can check that the degree $d(x)$ of any node $x = (\{v_0,v_1\},b) \in V^*$ is at most $d(v_0) + d(v_1)$. Define the transition matrix
\[
P^*\left( x,y \right) = \begin{cases}
 1/m & \textrm{if } \{x,y\} \in E^* \\
 1-d(x)/m & \textrm{if } x=y \\
 0 & \textrm{otherwise.}
 \end{cases}
\]
Consider an arbitrary initial configuration and two tokens located at nodes $v_0$ and $v_1$ of $G$. The expected meeting time of these two tokens is the same as the expected time to reach from $(v_0,v_1,0) \in V^*_0$ to any node in $V^*_1$ by the random walk given by $P^*$.

Observe that $G^*$ has $\Theta(n^2)$ vertices and $O(nm)$ edges. The first claim is immediate. For the second claim observe that
\begin{align*}
  |E^*| &= \frac{1}{2} \cdot \sum_{x \in V^*} d(x) \\
  &\le \sum_{ \{ u,v \} \subseteq V : u \neq v } (d(u) + d(v)) \le (n-1) \left( \sum_{u} d(u) + \sum_{v} d(v) \right) \in O(mn).
\end{align*}
Next note that $G^*$ satisfies $\diam(G^*) \in O( \diam(G) )$, since we can move the two tokens from any two distinct vertices to any other pair of distinct vertices using $O(\diam(G))$ transitions.

Finally, note that the Markov chain is irreducible. By the same arguments as in \lemmaref{lemma:stationary-pp-walk} and \lemmaref{lemma:hitting-time}, we get that the hitting time is at most $O(\diam(G)\cdot n^3 m)$, since $|V^*| \in \Theta(n^2)$, $|E^*| \in \Theta(nm)$, and the diameter of $G^*$ is $\Theta(\diam(G))$.
\end{proof}

Fix an arbitrary constant $c \ge 2$ and set
\[
T = \left \lceil 2  \cdot \max\{ H_\textrm{max}, M_\textrm{max}\} \right \rceil \qquad R = \lceil (c+2) \log n \rceil  \qquad T^* = R T.
\]
We show that in $T^*$ steps all pairs of tokens have met with high probability. 

\begin{remark}
For any $0 < p < 1$, the following inequality holds:
\[
\sum_{k=0}^\infty (k+1) p^k = \frac{1}{(p-1)^2}.
\]
\end{remark}

\begin{remark}
Let $A_0, \ldots, A_R$ be events. Then
\[
\Pr\left[ \bigcap_{i=0}^R A_i \right] = \Pr[A_0] \cdot \prod_{i=1}^k \Pr[A_i \mid A_{i-1}].
\]
\end{remark}

\begin{lemma}\label{lemma:m-bound}
We have $\Pr[M \ge T^*] \le 1/n^c$ and $\E[M] \le 4T^*$.
\end{lemma}
\begin{proof}
Let $M_t(a,b)$ be the first meeting time between tokens $a$ and $b$ after $t$ steps and $M_t = \max_{a,b} M_t(a,b)$. Note that $M_0(a,b) = M(a,b)$ and $M_0 = M$. First, we show that for any $a,b \in [n]$ and $t \ge 0$, the inequality 
\[
\Pr[M_t(a,b) \ge T^*] \le 1/n^{c+2}
\]
holds. For $a=b$ the claim is vacuous, so assume $a \neq b$. By Markov's inequality, the probability that tokens $a$ and $b$ do not meet within $T-1$ steps starting from any configuration $x_t$ is 
\[
\Pr[M_t(a,b) \ge T] \le \frac{\E[M_t(a,b)]}{T} \le \frac{M_\textrm{max}}{T} \le \frac{1}{2}.
\]
By repeating the experiment $R$ times, we observe that the probability that the tokens $a$ and $b$ do not meet within $T^* = RT$ steps (starting from any configuration $x_t$) satisfies
\begin{align*}
  \Pr[M_t(a,b) \ge RT] &= \Pr[ M_t(a, b) \ge T ] \cdot \prod_{r=1}^{R-1} \Pr[ M_{t+Tr}(a, b) \ge T \mid M_{t+T(r-1)}(a,b) \ge T ] \\
&\le  \prod_{r=1}^R \frac{1}{2} = \left(\frac{1}{2}\right)^{\lceil (c+2) \log n \rceil} \le \frac{1}{n^{c+2}}.
\end{align*}
Finally, we can bound the probability that any pair of tokens fails to meet before $T^*$ steps by applying the union bound:
\[
\Pr[M_t \ge T^*] \le \sum_{a,b} \Pr[ M_t(a,b) \ge T^* ] \le \sum_{a,b} \frac{1}{n^{c+2}} \le \frac{n^2}{n^{c+2}} = \frac{1}{n^c}. 
\]
Observe that for any $k > 0$ we have 
\begin{align*}
  \Pr[M_t \ge kT^*]
&=\Pr[M_t \ge T^*] \cdot \prod_{i=1}^{k-1} \Pr[M_{t+iT^*} \ge T^* \mid M_{t+(i-1)T^*} \ge T^*] 
                \le \prod_{i=1}^{k} \frac{1}{n^c} \le \frac{1}{2^k}, 
\end{align*}
since $n \ge 2$ and $c \ge 1$. To bound the expectation observe that 
\begin{align*}
\E[M] = \E[M_0] &\le \sum_{k=0}^{\infty} (k+1)T^* \cdot \Pr[M_0 \ge kT^*] \\
      &\le T^* \sum_{k=0}^{\infty} \frac{k+1}{2^{k}} = 4T^*. \qedhere            
\end{align*}
\end{proof}

\subsection{Stabilisation time of the token-based protocol}

We analyse the dynamics of the token-based leader election protocol. Let $C$ be the number of time steps until a single black token remains.

\begin{lemma}
The random variable $C$ satisfies $\Pr[C \ge T^*] \ge 1/n^c$.\label{lemma:c-bound}
 \end{lemma}
 \begin{proof}
 Observe that $C$ corresponds to the time when the last pair of black tokens meet. Now 
\[
C = \max \{ M(a,b) : a \neq b \} = M
\]
and the claim follows from \lemmaref{lemma:m-bound}.
\end{proof}

Let $L$ be the stabilisation time of the protocol, that is, the time until there is exactly one leader candidate remaining. Recall that a leader candidate becomes a follower if it receives a white token from some other node and a follower never becomes a candidate again. Thus, a node is a leader candidate at step $t$ if and only if it has not been hit by a white token. Whenever a white token hits a leader candidate, the token becomes inactive. This ensures that a single leader is always elected, as there will be exactly $n-1$ white tokens created during the execution of the protocol.

\begin{lemma}\label{lemma:hitting-time-bound}
Let $u$ and $v$ be distinct nodes. The probability that a token starting from $u$ does not hit $v$ within $T^*$ steps it at most $1/n^{c+2}$.
\end{lemma}
\begin{proof}
By \lemmaref{lemma:hitting-time} and Markov's inequality, the probability that a token starting from $u$ does not hit node $v$ in $T$ steps is bounded by
\[
 \frac{H(u,v)}{T} \le  \frac{H_\textrm{max}}{2 H_\textrm{max}} \le 1/2.
\]
Again, by repeating experiment for $R$ times, we get that with probability at most $2^{-R} \le 1/n^{c+2}$ the token starting from $u$ does not hit $v$ in $RT = T^*$ steps.
\end{proof}

\begin{lemma}
The random variable $L$ satisfies $\Pr[L < 2T^*] \le 2/n^c$ and $\E[L] \le 8T^*$.
\end{lemma}
\begin{proof}
Observe that conditioned on the event that there is only one black token remaining, the probability that some node $v$ does not become a follower is bounded by the event that node $v$ does not receive a white token. This is in turn bounded by the probability of the event that \emph{some} token does not visit $v$ within $T^*$ steps.
Hence,
\begin{align*}
 \Pr[L \ge t + T^* \mid C  < t ] &\le \sum_{v} \Pr[\textrm{node } v \textrm{ is not hit by a white token by time } t + T^* \mid C < t] \\
 &\le \sum_{v} \sum_a \Pr[\textrm{node } v \textrm{ is not hit by token } a \textrm{ by time } t + T^*] \\
 &\le \sum_v \sum_a \frac{1}{n^{c+2}} \\ 
 &\le \frac{1}{n^c},
\end{align*}
where in the second to last step we applied \lemmaref{lemma:hitting-time-bound} and in the last step the fact that there at most $n$ leader candidate nodes and $n$ tokens. By law of total probability, we get that
\begin{align*}
\Pr[L \ge 2T^*] &= \Pr[L \ge 2T^* \mid C < T^*] \cdot \Pr[C < T^*] + \Pr[L \ge 2T^* \mid C \ge T^*] \cdot \Pr[C \ge T^*] \\
&\le (1/n^c) \cdot (1-1/n^c) + 1/n^c  \\
&\le 2/n^c,
\end{align*}
where in the second to last step we apply the bound $\Pr[C \ge T^*] \le 1/n^c$ given by \lemmaref{lemma:c-bound}. 
Since $c \ge 2$ and considering repeated stabilisation attempts, we get that
\begin{align*}
  \E[L] &\le 2T^* \sum_{k=0}^{\infty} (k+1) \Pr[ L > 2kT^* ] \\
      &\le 2T^* \sum_{k=1}^{\infty} (k+1) \left(\frac{2}{n^c}\right)^k \\
      &\le 2T^* \sum_{k=1}^{\infty} (k+1) \left(\frac{1}{2}\right)^k \le 8T^*. \qedhere
\end{align*}
\end{proof}

\paragraph{Proof of \theoremref{thm:token-le}.}
\begin{proof}
The protocol uses only 6 states as each node only stores whether it is a leader and what is the type of its token. Since $T^* = O(\max \{ H_\textrm{max}, M_\textrm{max}\} \cdot\log n)$ and $H_\textrm{max} \in O(\diam(G) \cdot nm)$ by \lemmaref{lemma:hitting-time} and $M_\textrm{max} \in O(\diam(G) \cdot n^3 m )$ by \lemmaref{lemma:mmax-bound}. Thus, the protocol stabilises in $O(\diam(G) \cdot n^3 m \log n)$ steps w.h.p. and in expectation.
\end{proof}

\section{Conclusions}\label{sec:conclusions}

As our main result, we established a general framework for simulating clique-based protocols in arbitrary, connected regular graphs. We now conclude by briefly discussing some limitations of our approach and summarise key problems left open by this work:
\begin{itemize}
\item We assume that the nodes have access to a single random bit per interaction. The random bits are used only by the shuffling protocol of \sectionref{sec:simulation} to avoid technical parity issues arising in the mixing of the random walks on the symmetric group. It seems plausible that this assumption can be avoided, by exploiting the stochastic nature of the population protocol scheduler to e.g.\ generate synthetic coins~\cite{alistarh2018space-optimal} or to argue that these parity issues are avoided by the virtue of having a random number of shuffling steps.

\item We assume that in each interaction step in the population protocol model, one of the interacting nodes is assigned to be an initiator and the other a responder to provide elementary symmetry-breaking. This is again a common assumption in population protocol literature. The simulation framework uses this assumption only in the construction of the phase clock, where in certain situations ties need to be broken. It again seems plausible that this assumption can be avoided, but this would necessitate revisiting the involved graphical load balancing argument of Peres et al.~\cite{peres2014graphical} with different tie-breaking. 

\item We focus on regular interaction graphs. The justification for this assumption is two-fold. First, this assumption is only used once: in \sectionref{sec:clocks}, to obtain clean bounds for the skew of the phase clock. However, upon close inspection, we notice that this regularity assumption can be relaxed in many cases if the minimum and maximum degrees do not deviate too much from the average degree of the graph.  As \theoremref{thm:routability-stacked-ip-mixing} can be used to bound the mixing time of the interchange process in non-regular graphs as well, we can use our simulation framework to obtain fast leader election and exact majority algorithms also on some non-regular graphs. 
See Appendix~\ref{app:non-regular} for a formal statement and a concrete illustration.

Second, regular graphs are also justified by the fact that they provide an immediate extension of the notion of \emph{parallel time}: the expected number of interactions in any time interval is the same for all nodes, and prior work on this problem has naturally focused on them~\cite{draief2012-convergence, cooper2016-fast}.  Nevertheless, obtaining bounds for phase clocks and related load balancing processes in non-regular graphs remains an interesting open problem.

\item The simulation overhead has a polylogarithmic dependency on $n$. To simplify the  presentation, we have made no particular effort to optimise the degree of this polylogarithmic dependency. The dependency can be improved by providing better bounds on the $k$-stack interchange process. Indeed, even in the case of the well-studied (1-stack) interchange process, exact bounds on mixing time have been---and still remain---an open question for many graph classes~\cite{jonasson2012interchange}. Improved bounds for these processes imply better running time bounds for our simulations.

\item Our complexity bounds have a quadratic dependency on $d/\beta$. We conjecture a polynomial dependency on the expansion properties is necessary for step complexity and leave the investigation of tight space-time trade-offs for population protocols in the general graphical setting as an intriguing open problem. 
\end{itemize}

\subsection*{Acknowledgements}
We thank Giorgi Nadiradze for pointing out the generalisation of the phase clock construction to non-regular graphs. We also thank anonymous reviewers for their useful comments on earlier versions of this manuscript. This project has received funding from the European Research Council (ERC) under the European Union's Horizon 2020 research and innovation programme (grant agreement No.\ 805223 ScaleML), and from the European Union’s Horizon 2020 research and innovation programme under the Marie Sk{\l}odowska-Curie grant agreement No.\ 840605.

\bibliographystyle{plain}
\bibliography{population-protocols}

\begin{thebibliography}{10}

\bibitem{aldous1983random}
David Aldous.
\newblock Random walks on finite groups and rapidly mixing {M}arkov chains.
\newblock In {\em S{\'e}minaire de Probabilit{\'e}s XVII 1981/82}, pages
  243--297. Springer, 1983.

\bibitem{aldous-fill-2014}
David Aldous and James~Allen Fill.
\newblock Reversible markov chains and random walks on graphs, 2002.
\newblock Unfinished monograph, recompiled 2014, available at
  \url{http://www.stat.berkeley.edu/users/aldous/RWG/book.html}.

\bibitem{alistarh2017time}
Dan Alistarh, James Aspnes, David Eisenstat, Rati Gelashvili, and Ronald~L
  Rivest.
\newblock Time-space trade-offs in population protocols.
\newblock In {\em Proc.\ 28th Annual ACM-SIAM Symposium on Discrete Algorithms
  (SODA 2017)}, pages 2560--2579, 2017.

\bibitem{alistarh2018space-optimal}
Dan Alistarh, James Aspnes, and Rati Gelashvili.
\newblock Space-optimal majority in population protocols.
\newblock In {\em Proc.\ 29th ACM-SIAM Symposium on Discrete Algorithms (SODA
  2018)}. SIAM, 2018.

\bibitem{alistarhg15}
Dan Alistarh and Rati Gelashvili.
\newblock Polylogarithmic-time leader election in population protocols.
\newblock In {\em Proc.\ 42nd International Colloquim on Automata, Languages,
  and Programming (ICALP 2015)}, pages 479--491, 2015.

\bibitem{alistarh-survey}
Dan Alistarh and Rati Gelashvili.
\newblock Recent algorithmic advances in population protocols.
\newblock {\em {SIGACT} News}, 49(3):63--73, 2018.

\bibitem{alistarh2015-fast}
Dan Alistarh, Rati Gelashvili, and Milan Vojnovi{\'c}.
\newblock Fast and exact majority in population protocols.
\newblock In {\em Proc.\ 34th ACM Symposium on Principles of Distributed
  Computing (PODC 2015)}, pages 47--56, 2015.

\bibitem{angluin2006computation}
Dana Angluin, James Aspnes, Zo{\"e} Diamadi, Michael~J Fischer, and Ren{\'e}
  Peralta.
\newblock Computation in networks of passively mobile finite-state sensors.
\newblock {\em Distributed computing}, 18(4):235--253, 2006.

\bibitem{angluin2006stably}
Dana Angluin, James Aspnes, and David Eisenstat.
\newblock Stably computable predicates are semilinear.
\newblock In {\em Proc.\ 25th ACM Symposium on Principles of distributed
  computing (PODC 2006)}, pages 292--299, 2006.

\bibitem{angluin2008fast-computation}
Dana Angluin, James Aspnes, and David Eisenstat.
\newblock Fast computation by population protocols with a leader.
\newblock {\em Distributed Computing}, 21(3):183--199, 2008.

\bibitem{AAER07}
Dana Angluin, James Aspnes, David Eisenstat, and Eric Ruppert.
\newblock The computational power of population protocols.
\newblock {\em Distributed Computing}, 20(4):279--304, 2007.

\bibitem{angluin2008self}
Dana Angluin, James Aspnes, Michael~J Fischer, and Hong Jiang.
\newblock Self-stabilizing population protocols.
\newblock {\em ACM Transactions on Autonomous and Adaptive Systems (TAAS)},
  3(4):1--28, 2008.

\bibitem{aspnes2009introduction}
James Aspnes and Eric Ruppert.
\newblock An introduction to population protocols.
\newblock In {\em Middleware for Network Eccentric and Mobile Applications},
  pages 97--120. Springer, 2009.

\bibitem{avin2017distributed}
Chen Avin, Michael Borokhovich, Zvi Lotker, and David Peleg.
\newblock Distributed computing on core--periphery networks: Axiom-based
  design.
\newblock {\em Journal of Parallel and Distributed Computing}, 99:51--67, 2017.

\bibitem{azar1999balanced}
Yossi Azar, Andrei~Z. Broder, Anna~R. Karlin, and Eli Upfal.
\newblock Balanced allocations.
\newblock {\em SIAM Journal on Computing}, 29(1):180--200, 1999.

\bibitem{beauquier2013self}
Joffroy Beauquier, Peva Blanchard, and Janna Burman.
\newblock Self-stabilizing leader election in population protocols over
  arbitrary communication graphs.
\newblock In {\em International Conference on Principles of Distributed
  Systems}, pages 38--52. Springer, 2013.

\bibitem{BEFKKR18}
Petra Berenbrink, Robert Els{\"{a}}sser, Tom Friedetzky, Dominik Kaaser, Peter
  Kling, and Tomasz Radzik.
\newblock A population protocol for exact majority with {$O(\log_{5/3} n)$}
  stabilization time and {$\Theta(\log n)$} states.
\newblock In {\em Proc.\ 32nd International Symposium on Distributed Computing
  (DISC 2018)}, pages 10:1--10:18, 2018.

\bibitem{BFKMW16}
Petra Berenbrink, Tom Friedetzky, Peter Kling, Frederik Mallmann-Trenn, and
  Chris Wastell.
\newblock Plurality consensus in arbitrary graphs: Lessons learned from load
  balancing.
\newblock In {\em Proc.\ 24th Annual European Symposium on Algorithms (ESA
  2016)}, volume~57, pages 10:1--10:18, 2016.

\bibitem{berenbrink2018tight}
Petra Berenbrink, George Giakkoupis, and Peter Kling.
\newblock Tight bounds for coalescing-branching random walks on regular graphs.
\newblock In {\em Proceedings of the Twenty-Ninth Annual ACM-SIAM Symposium on
  Discrete Algorithms}, pages 1715--1733. SIAM, 2018.

\bibitem{berenbrink2020optimal}
Petra Berenbrink, George Giakkoupis, and Peter Kling.
\newblock Optimal time and space leader election in population protocols.
\newblock In {\em Proc.\ 52nd Annual ACM SIGACT Symposium on Theory of
  Computing (STOC 2020)}, pages 119--129, 2020.

\bibitem{BKKO18}
Petra Berenbrink, Dominik Kaaser, Peter Kling, and Lena Otterbach.
\newblock Simple and efficient leader election.
\newblock In {\em Proc.\ 1st Symposium on Simplicity in Algorithms (SOSA
  2018)}, pages 9:1--9:11, 2018.

\bibitem{blondin2018large}
Michael Blondin, Javier Esparza, and Stefan Jaax.
\newblock Large flocks of small birds: on the minimal size of population
  protocols.
\newblock In {\em Proc.\ 35th Symposium on Theoretical Aspects of Computer
  Science (STACS 2018)}. Schloss Dagstuhl-Leibniz-Zentrum fuer Informatik,
  2018.

\bibitem{brijder2019computing}
Robert Brijder.
\newblock Computing with chemical reaction networks: a tutorial.
\newblock {\em Natural Computing}, 18(1):119--137, 2019.

\bibitem{caputo2010aldous}
Pietro Caputo, Thomas~M. Liggett, and Thomas Richthammer.
\newblock Proof of {A}ldous' spectral gap conjecture.
\newblock {\em Journal of the American Mathematical Society}, 23(3):831--851,
  2010.

\bibitem{CMNPS11}
Ioannis Chatzigiannakis, Othon Michail, Stavros Nikolaou, Andreas Pavlogiannis,
  and Paul~G Spirakis.
\newblock Passively mobile communicating machines that use restricted space.
\newblock In {\em Proc.\ 7th ACM SIGACT/SIGMOBILE International Workshop on
  Foundations of Mobile Computing}, pages 6--15, 2011.

\bibitem{chen2019self}
Hsueh-Ping Chen and Ho-Lin Chen.
\newblock Self-stabilizing leader election.
\newblock In {\em Proceedings of the 2019 ACM Symposium on Principles of
  Distributed Computing}, pages 53--59, 2019.

\bibitem{chen2020ssle}
Hsueh-Ping Chen and Ho-Lin Chen.
\newblock Self-stabilizing leader election in regular graphs.
\newblock In {\em Proceedings of the 39th Symposium on Principles of
  Distributed Computing}, PODC '20, page 210–217, New York, NY, USA, 2020.
  Association for Computing Machinery.

\bibitem{cooper2013coalescing}
Colin Cooper, Robert Elsa\"sser, Hirotaka Ono, and Tomasz Radzik.
\newblock Coalescing random walks and voting on connected graphs.
\newblock {\em SIAM Journal on Discrete Mathematics}, 27(4):1748--1758, 2013.

\bibitem{cooper2016-fast}
Colin Cooper, Tomasz Radzik, Nicol{\'a}s Rivera, and Takeharu Shiraga.
\newblock Fast plurality consensus in regular expanders.
\newblock In {\em Proc.\ 31st International Symposium on Distributed Computing
  (DISC 2017)}, pages 13:1--13:16, 2017.

\bibitem{diaconis1993comparison}
Persi Diaconis and Laurent Saloff-Coste.
\newblock Comparison techniques for random walk on finite groups.
\newblock {\em The Annals of Probability}, 21(4):2131--2156, 1993.

\bibitem{diaconis1981generating}
Persi Diaconis and Mehrdad Shahshahani.
\newblock Generating a random permutation with random transpositions.
\newblock {\em Zeitschrift f{\"u}r Wahrscheinlichkeitstheorie und verwandte
  Gebiete}, 57(2):159--179, 1981.

\bibitem{dieker2010interlacings}
AB~Dieker.
\newblock Interlacings for random walks on weighted graphs and the interchange
  process.
\newblock {\em SIAM Journal on Discrete Mathematics}, 24(1):191--206, 2010.

\bibitem{doty2020majority}
David Doty, Mahsa Eftekhari, and Eric Severson.
\newblock A stable majority population protocol using logarithmic time and
  states, 2020.
\newblock arXiv:2012.15800.

\bibitem{doty2018stable}
David Doty and David Soloveichik.
\newblock Stable leader election in population protocols requires linear time.
\newblock {\em Distributed Computing}, 31(4):257--271, 2018.

\bibitem{draief2012-convergence}
Moez Draief and Milan Vojnovi{\'c}.
\newblock Convergence speed of binary interval consensus.
\newblock {\em SIAM Journal on Control and Optimization}, 50(3):1087--1109,
  2012.

\bibitem{elsaesser-survey}
Robert Els{\"{a}}sser and Tomasz Radzik.
\newblock Recent results in population protocols for exact majority and leader
  election.
\newblock {\em Bulletin of the {EATCS}}, 126, 2018.

\bibitem{ghaffari2017distributed}
Mohsen Ghaffari, Fabian Kuhn, and Hsin-Hao Su.
\newblock Distributed {MST} and routing in almost mixing time.
\newblock In {\em Proceedings of the ACM Symposium on Principles of Distributed
  Computing}, pages 131--140, 2017.

\bibitem{gasieniec2018fast}
Leszek G\k{a}siniec and Grzegorz Stachowiak.
\newblock Fast space optimal leader election in population protocols.
\newblock In {\em Proc.\ 29th ACM-SIAM Symposium on Discrete Algorithms (SODA
  2018)}, 2018.

\bibitem{gasieniec2018almost}
Leszek G\k{a}siniec, Grzegorz Stachowiak, and Przemysl{}aw Uzna\'nski.
\newblock Almost logarithmic-time space optimal leader election in population
  protocols.
\newblock In {\em Proc.\ 31st ACM Symposium on Parallelism in Algorithms and
  Architectures (SPAA 2019)}, 2019.

\bibitem{gasieniec2020time}
Leszek G\k{a}siniec, Grzegorz Stachowiak, and Przemysl{}aw Uzna\'nski.
\newblock Time and space optimal exact majority population protocols, 2020.

\bibitem{jonasson2012interchange}
Johan Jonasson.
\newblock Mixing times for the interchange process.
\newblock {\em Latin American Journal of Probability and Mathematical
  Statistics}, 9(2):667--683, 2012.

\bibitem{leighton1999multicommodity}
Tom Leighton and Satish Rao.
\newblock Multicommodity max-flow min-cut theorems and their use in designing
  approximation algorithms.
\newblock {\em Journal of the ACM}, 46(6):787--832, 1999.

\bibitem{levin2017markov}
David~A. Levin and Yuval Peres.
\newblock {\em Markov Chains and Mixing Times}.
\newblock American Mathematical Society, 2 edition, 2017.

\bibitem{lovasz1993random}
L{\'a}szl{\'o} Lov{\'a}sz.
\newblock Random walks on graphs: A survey.
\newblock {\em Combinatorics, Paul Erd\"os is Eighty}, 2(1):1--46, 1993.

\bibitem{MNRS14}
George~B. Mertzios, Sotiris~E. Nikoletseas, Christoforos Raptopoulos, and
  Paul~G. Spirakis.
\newblock Determining majority in networks with local interactions and very
  small local memory.
\newblock In {\em Proc.\ 41st International Colloquium on Automata, Languages,
  and Programming (ICALP 2014)}, pages 871--882, 2014.

\bibitem{mertzios2017determining}
George~B Mertzios, Sotiris~E Nikoletseas, Christoforos~L Raptopoulos, and
  Paul~G Spirakis.
\newblock Determining majority in networks with local interactions and very
  small local memory.
\newblock {\em Distributed Computing}, 30(1):1--16, 2017.

\bibitem{oliveira2013mixing}
Roberto~Imbuzeiro Oliveira.
\newblock Mixing of the symmetric exclusion processes in terms of the
  corresponding single-particle random walk.
\newblock {\em The Annals of Probability}, 41(2):871--913, 2013.

\bibitem{peres2014graphical}
Yuval Peres, Kunal Talwar, and Udi Wieder.
\newblock Graphical balanced allocations and the $(1+\beta)$-choice process.
\newblock {\em Random Structures and Algorithms}, 47(4):760--775, 2014.

\bibitem{SS}
Thomas Sauerwald and He~Sun.
\newblock Tight bounds for randomized load balancing on arbitrary network
  topologies.
\newblock In {\em Proc.\ 53rd Annual {IEEE} Symposium on Foundations of
  Computer Science (FOCS 2012)}, pages 341--350, 2012.

\bibitem{sudo2020leader}
Yuichi Sudo and Toshimitsu Masuzawa.
\newblock Leader election requires logarithmic time in population protocols.
\newblock {\em Parallel Processing Letters}, 30(01):2050005, 2020.

\bibitem{sudo2019logarithmic}
Yuichi Sudo, Fukuhito Ooshita, Taisuke Izumi, Hirotsugu Kakugawa, and
  Toshimitsu Masuzawa.
\newblock Logarithmic expected-time leader election in population protocol
  model.
\newblock In {\em Proc.\ International Symposium on Stabilizing, Safety, and
  Security of Distributed Systems (SSS 2019)}, pages 323--337, 2019.

\bibitem{turing1990chemical}
Alan~M. Turing.
\newblock The chemical basis of morphogenesis.
\newblock {\em Philosophical Transactions of the Royal Society of London,
  Series B}, 237(641):37--72, 1952.

\bibitem{wilson2004mixing}
David~B. Wilson.
\newblock Mixing times of lozenge tiling and card shuffling {M}arkov chains.
\newblock {\em The Annals of Applied Probability}, 14(1):274--325, 2004.

\bibitem{yokota2020time}
Daisuke Yokota, Yuichi Sudo, and Toshimitsu Masuzawa.
\newblock Time-optimal self-stabilizing leader election on rings in population
  protocols.
\newblock In {\em International Symposium on Stabilizing, Safety, and Security
  of Distributed Systems}, pages 301--316. Springer, 2020.

\end{thebibliography}

\newpage

\appendix

\section{Analysis of the $k$-stack interchange process}\label{apx:interchange}

In this section, we will give a slightly more refined analysis of the mixing time bound for the $k$-stack interchange process provided by \theoremref{thm:d-beta-mixing}.

\subsection{Mixing time bound via low-congestion and low-dilation routings}

\paragraph{Routings on graphs.} 
A \emph{routing} on a graph $G$ is a map $f$ that takes every pair of vertices onto a path $f(u,v)$ in $G$ connecting the vertices $u$ and $v$. 
The \emph{congestion} of a routing is the maximum number of paths on which any edge appears and its
\emph{dilation} is the length of the longest path in the routing. A graph $G$ is $(C,D)$-routable if there exists a routing with congestion $C$ and dilation $D$.

\paragraph{The mixing time bound.}
We now aim to establish the following result bounds the mixing time of the $k$-stack interchange process on a $(C, D)$-routable graph. 

\begin{theorem}\label{thm:routability-stacked-ip-mixing}
  Suppose $G$ is $(C,D)$-routable graph with $n$ vertices and $m$ edges.
  For any constant $k > 0$, the $k$-stack interchange process has mixing time
  \[
  O\left( n \log n \cdot \max \left\{ \frac{CDm}{n^2}, D \right\} \right).
  \]
\end{theorem}

Before we give a proof of the above theorem, we note that \theoremref{thm:d-beta-mixing} is straightforward corollary of the above theorem and the next lemma, which follows from the results of Leighton and Rao~\cite[Theorem 18]{leighton1999multicommodity}.

\begin{lemma}\label{lemma:bounded-degree-routing}
  If $G$ is a $d$-regular graph with edge expansion $\beta$, then $G$ is $(C,D)$-routable with parameters
  \[
  C \in O\left(\frac{ n \log n}{\beta} \right) \quad \textrm{ and } \quad D \in O\left(\frac{d \log n}{\beta}\right).
  \]
\end{lemma}

\subsection{Preliminaries: The path comparison method}

Our analysis is based on the comparison method developed by Diaconis and Saloff-Coste~\cite{diaconis1993comparison}.
Specifically, let $\mu$ and $\tilde{\mu}$ be increment distributions whose supports $H$ and $\tilde{H}$ generate the symmetric group $S_N$. For each element $a \in \tilde{H}$, choose a representation $a = x_1 \cdots x_k$, where $k$ is odd and $x_i \in H$ for $1 \le i \le k$. Let $|a| = k$ and $R(x,a)$ denote the number of times $x$ appears in the representation of $a$.
Define
\[
B(x) = \frac{1}{\mu(x)} \sum_{a \in S_N} \tilde{\mu}(a) R(x,a) |a| \quad \textrm{ and } \quad B = \max_{x \in H} B(x).
\]
The quantity $B$ is called the \emph{bottleneck ratio} of the representation and it is useful for bounding $\ell^2$-distance to stationarity. We use the following special case of a lemma from \cite[Lemma 5]{diaconis1993comparison}.

\begin{lemma}[\cite{diaconis1993comparison}]\label{lemma:d2-distances}
  Let $\mu$ and $\tilde{\mu}$ be symmetric increment distributions that generate the symmetric group $S_N$.
  Let $B$ be the bottleneck ratio for a representation as defined above.
  Then
  \[
d_2^2(t) \le N! \cdot \exp\left({-\frac{t}{B}} \right) + \widetilde{d_2^2} \left( \left \lfloor \frac{t}{2B} \right\rfloor \right).
  \]
\end{lemma}

\paragraph{Random transpositions shuffle.}
The \emph{random transpositions shuffle} is a random walk on the symmetric group given by the increment distribution
\[
\mu(x) = \begin{cases}
  1/N & \textrm{if } x = \id \\
  2/N^2 & \textrm{if } x = (i~j) \\
  0 & \textrm{otherwise.}
\end{cases}
\]
We compare the $k$-stack interchange process against the random transpositions shuffle, whose mixing behaviour is well-understood:
Diaconis and Shahshahani~\cite{diaconis1981generating} gave the following bound on the mixing time of the random transpositions shuffle.

\begin{lemma}\label{lemma:random-transpositions-d2}
  Let $\mu$ be the increment distribution for the random transpositions shuffle on $S_N$. There exists a universal constant $C$ such that for any $c \ge 0$ and $t = \lfloor N (\log N + c) \rfloor$, 
  \[
 d^2_2(t) \le C e^{-2c}.
  \]
\end{lemma}

\subsection{Proof of Theorem \ref*{thm:routability-stacked-ip-mixing}}

We are now ready to give the proof of \theoremref{thm:routability-stacked-ip-mixing}, which we break into parts.

\paragraph{Increment distribution for the $k$-stack interchange process.}
Label the cards (tokens) from $\{0, \ldots, nk-1\}$ and write $u_i = uk+i$ for $u \in [n]$ and $i \in [k]$ so that $u_i$ denotes the $i$th card of node $u \in V$. Thus, $u_0$ is the top card and $u_{k-1}$ is the bottom card on the stack located at node $u$. Let $\sigma_u = (u_0 ~ u_{k-1} ~ u_{k-2} \ldots u_1)$ denote the permutation which moves the top card of $u$ to the bottom of its stack. Recall that $(u_0~v_0)$ is the transposition along an edge for neighbouring $u \neq v$. The increment distribution $\mu$ of the random walk is given by
\[
\mu(x) = \begin{cases}
  1/(2n) & \textrm{if } x = \sigma_u \textrm{ for some } u \in V \\
  1/(4m) & \textrm{if } x = (u_0~v_0) \textrm{ for some } \{u,v\} \in E \\
  1/4   & \textrm{if } x = \id \\
  0 & \textrm{otherwise.}
\end{cases}
\]
Thus, the support of the increment distribution of the $k$-stack interchange process is
\[
H = \{ \id \} \cup \{ \sigma_u : u \in V \} \cup \{ (u_0 ~ v_0) : \{u,v\} \in E \}.
\]

\paragraph{Comparison with random transpositions shuffle.}
  Let $\tilde{\mu}$ be the increment distribution of the random transpositions shuffle on $S_N$ and $\tilde{H} = \{ (u_i~v_j) : u,v \in V, i,j \in [k]\}$ the support of $\tilde{\mu}$. We start by choosing a representation for each transposition $(u_i~v_j) \in \tilde{H}$ in terms of an odd number of elements in $H$. After this, we bound the bottleneck ratio $B$ of the chosen set of representations and apply \lemmaref{lemma:d2-distances} to bound the mixing time.

\paragraph{Bounding the bottleneck ratio.} Consider the partition $\tilde{H} = \{ \id \} \cup \tilde{H}_1 \cup \tilde{H}_2$, where
  \[
\tilde{H}_1 = \{ (u_i~u_j) : u \in V, 0 \le i < j < k \} \quad \textrm{ and } \quad \tilde{H}_2 = \{ (u_i~v_j) : u,v \in V, u \neq v, 0 \le i < j < k  \}.
\]
For each $a \in \tilde{H}$ we find a representation in terms of elements in $H$ depending on which of the three parts $a$ belongs to. The identity element of $\tilde{H}$ can be represented by the identity of $H$, so we need to only find odd-length representations for elements in $\tilde{H}_1$ and $\tilde{H}_2$:
\begin{itemize}
  \item Suppose $a = (u_i~u_j) \in \tilde{H}_1$. We can represent this as
\[
(u_i ~ u_j) = \sigma_u^{k-i} \cdot (u_0~v_0) \cdot \sigma_u^{k-j+i} \cdot (u_0 ~ v_0) \cdot \sigma_u^{j-i} \cdot (u_0 ~ v_0) \cdot \sigma_u^{i},
\]
where $\sigma_u^i$ stands for $i$ repetitions of $\sigma_i$ and $v$ is a fixed neighbour of $u$ in $G$. The length $|a|$ of the representation is $2k+3$ and $R(x,a) \le \max\{ 2k, 3 \} \le 2k+1$ for each $x \in H$.

\item Suppose $a = (u_i~v_j) \in \tilde{H}_2$. Since the graph $G$ is $(C,D)$-routable, there exists a path $u = w^0, \ldots, w^\ell = v$ of length $1 \le  \ell \le D$ connecting $u$ and $v$ in $G$. Let
\[
 \rho_{uv} = (w^0_0 ~ w^1_0) \cdots (w^{\ell-1}_0 ~ w^{\ell}_0) \cdot (w_0^{\ell-1}~w_0^{\ell-2}) \cdots (w_0^1 ~ w_0^0)
 \]
 be the sequence of $2\ell - 1$ transpositions. With this, we can represent $a = (u_i~v_j)$ as
 \[
(u_i~v_j) = \sigma_u^{k-i} \cdot \sigma_v^{k-j} \cdot \rho_{uv} \cdot \sigma_v^{j} \cdot \sigma_u^i.
 \]
 Now $|(u~v_j)| = 2k+2\ell - 1 \le 2(k + D) - 1$. Note that $\sigma_u$ and $\sigma_v$ are used both at most $k$ times and $\rho_{uv}$ uses each edge-wise transposition at most twice. Hence $R(x,a) \le 2k$.
\end{itemize}
Next, we bound $B(x)$ for each $x \in H$. There are again three cases to consider:
\begin{itemize}
  \item Suppose $x = \id$. Since the identity element is only used to represent itself, we have that $B(\id) = \tilde{\mu}(\id)/\mu(\id) = 4/kn  \in O(1)$.

\item Suppose $x = \sigma_u \in H$ for some $u$. Note that $\mu(\sigma_u) = 1/(4n)$ and $\sigma_u$ appears in $k^2-1$ representations of elements in $\tilde{H}_1$  and $(k^2-1)n$ representations of elements in $\tilde{H}_2$. Thus,
\begin{align*}
  B(x) &= \frac{1}{\mu(x)} \sum_{a \in S_N} \tilde{\mu(a)} R(x,a) |a| = \frac{8}{nk^2} \left[ \sum_{a \in \tilde{H}_1} R(x,a) |a| + \sum_{a \in \tilde{H}_1} R(x,a) |a| \right].
\end{align*}
The sum over $\tilde{H}_1$ is bounded by $O(k^4)$ and the sum over $\tilde{H}_2$ by $O(nk^4 + nk^3 D)$. Hence, $B(x) \in O(k^2 + Dk)$.

\item Suppose $x = (u_0~v_0) \in H$ for some $\{u,v\} \in E$. Note that $\mu(\sigma_u) = 1/(4m)$ and $(u_0~v_0)$ is used in at most $2k^2$ representations in $\tilde{H}_1$ and in at most $C$ representations in $\tilde{H}_2$. Hence,
\begin{align*}
  B(x) &= \frac{1}{\mu(x)} \sum_{a \in S_N} \tilde{\mu(a)} R(x,a) |a| = \frac{8m}{(nk)^2} \left[ \sum_{a \in \tilde{H}_1} R(x,a) |a| + \sum_{a \in \tilde{H}_1} R(x,a) |a| \right]. 
\end{align*}
The sum over $\tilde{H}_1$ is bounded by $O(k^4)$ and the sum over $\tilde{H}_2$ by $O(Ck^2 + CDk)$. Thus, $B(x) \in O( CDmk^2/n^2)$.
\end{itemize}
Therefore, the bottleneck ratio 
$B = \max_{x \in H} B(x)$ is bounded by $O\left( k^2 \cdot \max\left\{ CDm/ n^2, D \right\} \right)$.

\paragraph{Bounding the mixing time.}
We can now bound the mixing time of the $k$-stack interchange process using the bound on the bottleneck ratio $B$ and \lemmaref{lemma:d2-distances}. Note that we can choose $t \in \Theta(B N \log N)$ so that the following inequalities hold:
\[
\widetilde{d^2_2}\left( \left \lfloor \frac{t}{2B} \right \rfloor \right) \le 1/32 \quad \textrm{ and } \quad \exp(-t/B) \le \frac{1}{32 N^N} \le \frac{1}{32 \cdot N!}.
\]
The first inequality is obtained using \lemmaref{lemma:random-transpositions-d2}. The total variation distance is bounded by \lemmaref{lemma:d2-distances} and the fact that $d_1(t) \le d_2(t)$ yielding
\begin{align*}
  \frac{1}{2}d_1(t) \le \frac{1}{2}d_2(t) &\le \frac{1}{2} \sqrt{ N! \cdot \exp\left({-\frac{t}{B}} \right) + \tilde{d}_2^2\left( \left \lfloor \frac{t}{2B} \right\rfloor \right) } \le \frac{1}{8}.
\end{align*}
The claim of \theoremref{thm:routability-stacked-ip-mixing} follows by observing that the mixing time is bounded by $O(B N \log N)$, which for constant $k$ is 
$O\left( n \log n \cdot \max\left\{ CDm/n^2, D \right\} \right)$,
as claimed.

\section{Complete analysis of decentralised phase clocks}\label{apx:process-gap}

Let $V$ be a collection of $n$ bins, which are initially empty, and let $\mu$ be a probability distribution on $V \times V$. Consider the process, where at every time step $t>0$, a pair $(u,v)$ is sampled according to $\mu$ and a ball is placed into the \emph{least} loaded of these two bins (in case of ties, place the ball into bin $u$). Let $\Delta^*(t)$ be the difference between the most and least loaded bins after step $t$.
Define $E(S)$ as the set of edges that have at least one end point in $S$, that is, 
\[
E(S) = \{ \{u,v\} \in E : \{u,v\} \cap S \neq \emptyset \}.
\]
The distribution $\mu$ on $V \times V$ is said to be $\delta$-expanding for $\delta>0$ if for all $S \subseteq V$ with $|S| \le n/2$ 
\begin{enumerate}[noitemsep]
\item $\mu(E(S)) \ge (1+\delta)|S|/{n}$, and
\item $\mu(E(S) \setminus \partial S) \le (1-\delta)|S|/n$
\end{enumerate}
hold, where $\partial S$ denotes the edge boundary of $S$. Peres et al.~\cite{peres2014graphical} showed that when the measure $\mu$ is well-behaved in the sense that it is $\delta$-expanding, then the gap is bounded by $O(\log n / \delta)$ at every step $t$, with high probability.

\begin{lemma}[\cite{peres2014graphical}]
  Let $\kappa > 0$ be a constant and $\mu$ be an $\delta$-expanding measure on $V \times V$. Then there exists a constant $c(\kappa)$ such that for any $t > 0$ the gap $\Delta^*(t)$ satisfies
  \[
  \Pr[ \Delta_\ell(t) > c(\kappa) \log n / \delta ] < t/n^\kappa.
  \]
\end{lemma}

The following observation establishes that the uniform distribution on edges of a regular, connected graph is always $\delta$-expanding for some $\delta > 0$. This in turn implies \lemmaref{lemma:unbounded-process-gap}.

\begin{lemma}
  Suppose $G$ is $d$-regular with edge expansion $\beta > 0$. The uniform distribution $\xi$ on the edges of $G$ is $(\beta/d)$-expanding.
\end{lemma}
\begin{proof}
  Let $S \subseteq V$ such that $|S| \le n/2$. Note that $|\partial S| \ge \beta |S|$, as the graph has edge expansion~$\beta$. Since the graph is $d$-regular, it has $m=nd/2$ edges and $|S|d/2 \le |E(S)| \le |S|d$. Now
  \begin{align*}
    \xi(E(S)) &= \frac{1}{m} |E(S)| \ge \frac{2}{nd} \left[ \sum_{v \in S} \left(\frac{d-\beta}{2} + \beta\right) \right] 
    = \frac{|S|}{n}\left(1 + \frac{\beta}{d}\right).
  \end{align*}
  This shows the first condition. For the second condition, observe that
  \begin{align*}
    \xi(E(S) \setminus \partial S) &= \frac{1}{m} |E(S) \setminus \partial(S)| \le \frac{2}{nd} \left[ \sum_{v \in S} \left(\frac{d-\beta}{2} \right) \right] = \frac{|S|}{n}\left(1 - \frac{\beta}{d}\right). \qedhere
  \end{align*}
\end{proof}

\subsection{Proof of Theorem~\ref{thm:clocks}}
\label{app:proof-theorem-clocks}

For convenience, we first restate the result. 

\clockthm* 
\begin{proof}
  Fix $\kappa, \gamma$ and $\phi$. We devise the clock protocol for $\ppmodel(G)$, where each node holds a counter value $c(v,t) \in [\phi]$. Each node initialises its clock value to $c(v,0) = 0$.
  Define
\[
  M_\phi(x,y) = \begin{cases}
    \max(x,y) & \textrm{if } |x-y| < \phi/2, \\
    \min(x,y) & \textrm{otherwise}.
  \end{cases}
  \]
  The clock protocol is now defined by the following update rule. When nodes $\{ u_0, u_1 \} \in E$ interact, where $u_0$ is the initiator and $u_1$ is the responder, they perform the following:
\begin{itemize}[noitemsep]
\item If $c(u_0,t) = c(u_1,t)$, then the initiator $u_0$ increments its clock value by one modulo $\phi$.
\item Otherwise, the node $u_i$ for which $M_\phi(c(u_0,t), c(u_1,t)) \neq c(u_i,t)$ holds increments its clock by one modulo $\phi$.
\end{itemize}
We argue that this protocol implements a $(\phi,\gamma,\kappa)$-clock, i.e., properties (1) and (2) are satisfied. Note that the above rules guarantee that in either case exactly one of the nodes increments its clock value by one modulo $\phi$. This implies property (2).

We show property (1) using a coupling argument.
Let $(X_t,Y_t)_{t \ge 0}$ be the following coupling of 
the unbounded balls-into-bins $(X_t)_{t \ge 0}$ and the bounded phase clock $(Y_t)_{t \ge 0}$ process.
For each directed edge $e_t = (u_0,u_1)$ sampled at step $t$, both of the processes are updated by applying their respective update rule to the pair $(u_0,u_1)$.

We show the following claim: Let $t \ge 0$. If $\Delta(t')=\Delta^*(t')$ and $\Delta^*(t') < \gamma$ hold for all $t' < t$, then in both processes the same bin is incremented in each step $t' \le t$ and $\Delta(t)=\Delta^*(t)$.
The base case $t=0$ is vacuous.
Suppose the claim holds for some $t > 0$ and $\Delta^*(t') < \gamma$ holds for all $t' < t$.
Let $(u_0,u_1) = e_{t+1}$  be the edge sampled at step $t+1$.
Observe that $c(u_i,t) = \ell(u_i,t) \bmod \phi$ for $i \in \{0,1\}$.
If $\ell(u_0,t)=\ell(u_1,t)$, then both the unbounded and unbounded process increments the bin of node $u_i$.
Otherwise, let $i \in \{0,1\}$ such that $\ell(u_i,t) > \ell(u_{1-i},t)$. The unbounded process increments the bin of node $u_i$. Now consider the bounded process. By assumption, we have that $\Delta(t)=\Delta^*(t) < \gamma \le \phi/2$ and so we have the following two cases:

\begin{itemize}
  \item If $c(u_i,t) < c(u_{1-i},t)$, then we have $|c(u_i,t)-c(u_{1-i},t)| > \phi/2$ and so the bounded process increments the bin $u_i$ by the update rule of the clock protocol.
\item In the second case, $c(u_i,t) > c(u_{1-i},t)$ we get that $|c(u_i,t)-c(u_{1-i},t)| < \phi/2$, and again, the bin $u_i$ is incremented by the update rule of the clock protocol.
\end{itemize}
Hence, in either case we have that $\Delta(t+1)=\Delta^*(t+1)$.
It now follows from \lemmaref{lemma:unbounded-process-gap} that property (2) is satisfied by the clock protocol.
\end{proof}

\subsection{Phase clocks for non-regular graphs}
\label{app:non-regular}

Finally, we observe that the phase clock construction is not restricted to only regular graphs. The uniform distribution on the edges of $G$ is $\delta$-expanding whenever (1) the minimum and maximum degree do not deviate too much from the average degree $\alpha = 2m/n$ and (2)~the expansion is sufficiently large compared to the average and minimum degree.

\begin{lemma}
  Let $G$ be a graph with minimum degree $d$, maximum degree $\Delta$, and average degree $\alpha$.  If $G$ satisfies
  \begin{enumerate}[label=(\alph*)]
  \item $\beta + d > \alpha$ and 
    \item $d + \Delta \le 2 \alpha$, 
  \end{enumerate}
  then the uniform distribution $\xi$ on the edges of $G$ is $\delta$-expanding for $\delta = (\beta+d)/\alpha - 1 > 0$.
\end{lemma}
\begin{proof}
  Let $S \subseteq V$ such that $|S| \le n/2$. Note that $|\partial S| \ge \beta |S|$, as the graph has edge expansion~$\beta$. Let $d \le \deg(v) \le \Delta$ denote the degree of node $v$ in $G$, $\outdeg(v,S)$ be the number neighbours $v$ has outside the set $S$, $\indeg(v,S)$ denote the number of neighbours $v$ has in the set $S$. First, observe that
  \begin{align*}
    \xi(E(S)) &= \frac{1}{m} |E(S)| \\
    &= \frac{1}{m} \left[ \sum_{v \in S} \left( \outdeg(v,S) + \frac{\indeg(v,S)}{2} \right) \right] \\ 
    &=  \frac{1}{m} \left[ \sum_{v \in S} \left( \outdeg(v,S) + \frac{\deg(v) - \outdeg(v,S)}{2} \right) \right] \\ 
    &= \frac{1}{2m} \left[ \sum_{v \in S} \left( \outdeg(v,S) + \deg(v) \right) \right] 
    \ge |S| \cdot \left( \frac{\beta + d}{2m} \right) \\
    &= \frac{|S|}{n} \cdot \left( \frac{\beta + d}{\alpha} \right) = (1+\delta) \cdot \frac{|S|}{n}.
  \end{align*}
  Thus, we have satisfied Condition (1) of a $\delta$-expanding measure. For the second condition, we note that
  \begin{align*}
    \xi(E(S) \setminus \partial S) &= \frac{1}{m} |E(S) \setminus \partial(S)| = \frac{1}{2m} \left[ \sum_{v \in S} \indeg(v, S) \right] \le \frac{1}{2m}\left[ \sum_{v \in S} \left( \deg(v) - \outdeg(v,S) \right) \right] \\
    &\le |S| \left(\frac{ \Delta - \beta}{2m}\right) = \frac{|S|}{n} \cdot  \left(\frac{ \Delta - \beta}{\alpha}\right) \le (1- \delta) \frac{|S|}{n}. \qedhere
  \end{align*}
\end{proof}

Note that one can also apply the construction on non-uniform probability distributions over $E$ (i.e.\ weighted graphs) as long as the distribution is $\delta$-expanding for some $\delta > 0$.

\paragraph{An example graph.}
For a simple non-regular graph that satisfies the above conditions, take a bipartite complete graph $K_{2r,2r}$ on $4r$ nodes and on one side add $r^2$ edges to form a complete bipartite subgraph on $r$ nodes. Now $n=4r$ and $m=4r^2 + r^2=5r^2$. One can check that the average degree is $\alpha = 2m/n = 5r/2$, minimum degree is $2r$, maximum degree is $3r$, and $\beta \ge 2r$. Thus, the uniform measure is $1/5$-expanding. 

\section{Additional details on the simulation theorem}

\subsection{Proof of an inequality}\label{apx:prod-upper-bound}
\begin{lemma}\label{lemma:prod-upper-bound}
  Let $a_i, b_i \in \mathbb{R}^+$ for $1 \le i \le t$. Then
  \[
  \left | \prod_{i=1}^t a_i - \prod_{j=1}^t b_j \right | \le \sum_{i=1}^t |a_i - b_i| \left( \prod_{k=1}^{i-1} a_k \right) \left( \prod_{h=i+1}^{t} b_h \right).
  \]
\end{lemma}
\begin{proof}
  For all $0 \le i \le t$ define
  \[
  c_i = \left( \prod_{k=1}^i a_k \right) \left( \prod_{h=i+1}^t b_h \right) \quad \textrm{ and } \quad d_i = (a_i-b_i)  \left( \prod_{k=1}^{i-1} a_k \right) \left( \prod_{h=i+1}^t b_h \right).
  \]
  The claim follows by observing that $d_i = c_i - c_{i-1}$ holds and
  \[
  \left | \prod_{i=1}^t a_i - \prod_{j=1}^t b_j  \right | = \left | c_t - c_0 \right| = \left | \sum_{i=1}^t (c_i - c_{i-1}) \right | \le \sum_{i=1}^t |c_i - c_{i-1}| = \sum_{i=1}^t |d_i|. \qedhere
  \]
\end{proof}

\subsection{Proof of Lemma~\ref{lemma:executions}}\label{apx:executions}

\begin{proof}
  Observe that each node $v$ updates the variables $a(v), b_0(v), \ldots, b_{k-1}(v)$ only during the steps $s(v,1), \ldots, s(v,R)$ when its local clock has reached the threshold value~$\vartheta$. \lemmaref{lemma:timing} implies that with high probability every node updates these variables for the $r$th time during the interval $\{ s_{\min}(r) + 1, \ldots, s_{\max}(r) + 1\}$ for any $1 \le r \le R$. In particular, this happens before step $t_{\min}(r+1)$, when the first node becomes receptive for the $(r+1)$th time. By letting $y_{r+1}(v_i)$ be the value of $b_i(v)$ after being updated for the $r$th time and $y'_{r+1} = y_r \circ \sigma_{r+1}$, we get that the configuration $x_{r+1}$ satisfies, with high probability,
\[
x_{r+1}(v) = f( x_{r}, y'_r(v_0, \ldots, v_{k-1})).
\]
Thus, the sequence given by $x_0, \ldots, x_{R}$ is an execution induced by the schedule $\sigma_1, \ldots, \sigma_R$.
\end{proof}

\section{Details on algorithms for the token shuffling model}\label{apx:token}

\subsection{Preliminaries}

 \begin{remark}
   We make use of the following elementary facts.
 \begin{itemize}[noitemsep]
 \item The law of total expectation: For random variables $X$ and $Y$ on the same probability space,
   \[
   \E[X] = \E[ E[X \mid Y ] ].
   \]
 \item Markov's inequality: For any nonnegative random variable $X$ and real value $a > 0$, \[
   \Pr[X \ge a] \le \frac{\E[X]}{a}.
 \]
 \item The union bound: For any events $A_0, \ldots, A_n$, we have
   \[
   \Pr\left[ \bigcup^{n}_{i=0} A_i \right] \le \sum_{i=0}^n \Pr[A_i].
   \]
 \end{itemize}
 \end{remark}

 We start with the following observation about a particular quadratic recurrence.

\begin{lemma}\label{lemma:recurrence}
 For any $n > x > 0$ and $r \ge 0$, the expression $g(r) = n(x/n)^{2^r}$ is the closed form solution of the quadratic recurrence
\[
g(r) = \begin{cases}
  g^2(r-1)/n & \textrm{if } r > 0 \\
  x & \textrm{otherwise.}
\end{cases}
\]
Moreover, $g(r) \le 1/n^\lambda$ holds for all $r \ge \log n + \log \ln n + \log(\lambda+1)$.
\end{lemma}
\begin{proof}
  We show the identity via induction.
  The base case $r=0$ is given by 
  $g(0) = n(x/n) = x$. For the inductive step, we have
  \[
g(r+1) = \frac{g^2(r)}{n} = \frac{1}{n}\left[n \left(\frac{x}{n}\right)^{2^r} \right]^2 = n \left(\frac{x}{n}\right)^{2 \cdot 2^r} = n \left(\frac{x}{n}\right)^{2^{r+1}}.  
\]
The second claim follows from the inequality $(1-1/n)^{nx} \le e^{-x}$, since
\[
g(r) \le n\left(1 - \frac{1}{n}\right)^{2^r} \le n\left(1 - \frac{1}{n}\right)^{(\lambda +1) n \ln n } \le n e^{-(\lambda +1) \ln n} \le 1/n^\lambda. \qedhere
\]
\end{proof}

\subsection{Proof of \lemmaref{lemma:consensus}}

\begin{proof}
  Let $U_r$ denote the set of nodes that have \emph{not} set their local state variable $a(\cdot)$ to the maximum input value after $r$ rounds. Fix a constant $\lambda > 0$ and let $2^R = (\lambda+1) n \ln n$. We prove the lemma by showing that for all $r \ge R$ the probability that $U_r$ is nonempty is at most $1/n^\lambda$.

  For each $v \in V$, let $Y_r(v)$ be the indicator variable for the event that node $v$ receives at least one token with the maximum input value in round $r$. By Step (3) of the protocol, if this event occurs, then $v$ sets $a(v)$ to the maximum input value at the end of round $r$. In particular, this implies that $v \notin U_{r'}$ for all $r' \ge r$.
  Note that $\Pr[Y_{r+1}(v) = 1 \mid U_r = b] \ge 1-b/n$.
  Define the random variable $X_r = |U_r|$ for each $r > 0$. Observe that by linearity of expectation, we have
\[
\E[X_{r+1} \mid X_r = b] = b - \E\left[\sum_{v \in U_r} Y_r(v)\right] \le b - \sum_{v \in U_r} \E[Y_r(v)] \le b^2/n.
\]
By law of total expectation, we get
\[
\E[X_{r+1}] = \E\left[ \E\left[X_{r+1} \mid X_{r} \right]\right] \le g(r),
\]
where $g(r)$ is the recurrence of \lemmaref{lemma:recurrence} with $x = |U_0| < n$. By Markov's inequality and the second claim of \lemmaref{lemma:recurrence}, we get
\begin{align*}
  \Pr[ X_R \ge 1] &\le \E[U_R] \le g(R) \le 1/n^\lambda. \qedhere
\end{align*}
\end{proof}

\subsection{Proof of Theorem~\ref{thm:token-le}}

\begin{proof}
  Assume that the broadcast protocol in Step (2) succeeds on each of the first $\Theta(\log n)$ iterations of the leader election protocol. This event happens with high probability.  Let $L_i$ be the (random) set of leader candidates after the $i$th iteration and $X_i = |L_i|$ be the random variable indicating the number of leader candidates after the $i$th iteration.

  For each leader candidate $v \in L_i$, let $B_i(v)$ be a random variable indicating whether node~$v$ had 1 as its $(i+1)$th coin flip.
  Note that $\E[B_i(v)] = 1/2$ for any $v \in L_i$.  Let $p(a)$ be the probability that each of the $|L_i|=a$ leader candidates have 0 as their coin flip. If this event occurs, then no leader candidate gets removed. Observe that $p(a) \le 1/2^a$. Assuming $a > 1$ and using linearity of expectation, we get 
  \[
\E\left[X_{i+1} \mid X_i = a \right] = a \cdot p(a) + (1-p(a)) \cdot \E\left[ \sum_{v \in L_i } B_i(v) \right]  = a \left[  p(a) + \frac{1-p(a)}{2} \right] \le \frac{3a}{4}.
\]
Thus, by law of total expectation $\E[X_t] \le n \left(3/4\right)^t$ holds. For any constant $\lambda >0$ we can set $t = \lambda \log_{4/3} n$. By Markov's inequality,
\[
\Pr\left[ X_t > 1 \right] \le \E\left[X_t \right]  \le n (3/4)^t \le 1/n^\lambda.
\]
Hence, with high probability, there remains only one leader candidate after $t \in O(\log n)$ iterations. As each iteration takes $\Theta(\log n)$ rounds, the algorithm stabilises in $O(\log^2 n)$ rounds with high probability, as desired. The state complexity bound comes from the fact that in Step (2) nodes count up to $\Theta(\log n)$ rounds. The algorithm uses two token types.
\end{proof}

\subsection{Proof of \theoremref{thm:majority}}

We trace the usual steps taken in analysing cancellation-doubling dynamics~\cite{angluin2008fast-computation,alistarh2018space-optimal,elsaesser-survey}.
Let $A_i$ and $B_i$ denote the number of \emph{majority} and \emph{minority} tokens after $i$ iterations, respectively. Define the discrepancy as $\Delta_i = A_i - B_i$ with $\Delta_0 > 0$ being the initial discrepancy. In addition, we use $A'_i$, $B'_i$, and $C_i$ to denote the number of majority, minority, and empty tokens after the $i$th cancellation phase.

The algorithm maintains the invariant $A_i > B_i$ with high probability: The cancellation rule removes exactly one majority and minority token each time it is applied. The doubling phase guarantees  that $A_{i+1} = 2A'_i$ and $B_{i+1} = 2B'_i$ holds with high probability. Finally, once $A_i = N$ holds, that is, all tokens have taken the majority value, then $A_{i+1} = N$ holds. This ensures that once all tokens are of the same type, they remain so for all subsequent rounds.

\begin{lemma}\label{lemma:cancellation}
  Let $i \ge 0$. If $A_i > B_i$ holds, then
  one of the following holds with high probability:
\begin{enumerate}[noitemsep]
\item $B_i'= 0$, that is, there are no minority tokens after the $i$th cancellation phase, or
\item $C_i \ge 3N/5$, that is, there are at least $3N/5$ empty tokens after the $i$th cancellation phase.
\end{enumerate}
\end{lemma}
\begin{proof}
  Observe that $B_i < A_i < N/5$ implies the second condition $C_i \ge 3N/5$. Thus, assume that $A_i \ge N/5$ holds. Consider the event that a fixed minority token $b$ is not cancelled during the $t$ rounds of the cancellation phase \emph{conditioned on} there being at least $N/5$ majority tokens in each round of this phase. The probability that $b$ meets a majority token is at least $1/5$. Hence, the probability that $b$ is not cancelled during any of these $t$ rounds is at most
\[
(1-1/5)^t = (4/5)^t = 1/N^{\lambda+1}.
\]
Now by taking the union bound over all $B'_i < N/2$ minority tokens yields that either (1) all minority tokens gets cancelled with probability at least $1-1/N^\lambda$ or (2) there are less than $N/5$ majority tokens remaining after the $i$th cancellation phase. This proves the lemma.
\end{proof}

\begin{lemma}\label{lemma:doubling}
If $C_i \ge 3N/5$ holds, then $\Delta_{i+1} = 2 \Delta_i$ holds with high probability.
\end{lemma}
\begin{proof}
  We say that a token of type $Z \in \{0,1\}$ \emph{splits} if it activates the rule $Z + \emptyset \to Z^{1/2} + Z^{1/2}$. Observe that after a nonempty token of type $Z$ splits during the $i$th doubling phase, it becomes a token of type $Z^{1/2}$ and cannot split again before the $(i+1)$th doubling phase.

  Recall that by assumption, there are $C_i \ge 3N/5$ empty tokens. Hence, the probability that a nonempty token splits in a single round of the doubling phase is at least $1/5$, since at most $2N/5$ nonempty tokens can split (each removing an empty token from the system). Therefore, the probability that a nonempty token does not split is at most
  \[
(1-1/5)^t = (4/5)^t = 1/N^{\lambda+1}.
  \]
  By union bound, the probability that some nonempty token does not split is at most $1/N^\lambda$. Now
  \[
\Delta_{i+1} = 2(A'_i - B_i') = 2(A_i - B_i) = 2\Delta_i. \qedhere
  \]
\end{proof}

\paragraph{Bounding the number of iterations.} Define the following two random variables
\[
 K_0 = \min \{ i : B_i = 0 \} \qquad \textrm{ and } \qquad K_1 = \min \{ i : A_i = N \}.
\]
Here $K_0$ indicates the iteration after which no tokens taking the minority value are present anymore. The variable $K_1$ is the first iteration when all tokens have been converted to the majority value. Note that $K_1 \ge K_0$ since even though no minority tokens are remaining, some empty tokens may remain after doubling phases.

\begin{lemma}\label{lemma:k0-bound}
  The random variable $K_0$ satisfies $K_0 \le \lceil \log N \rceil + 1$ with high probability.
\end{lemma}
\begin{proof}
  Let $K = \lceil \log N \rceil + 1$. 
  Suppose that $B_i > 0$ holds for all $0 \le i \le K$. By \lemmaref{lemma:cancellation}, we have that $C_i \ge 3N/5$ holds with high probability. By \lemmaref{lemma:doubling}, we get that $\Delta_{i+1} = 2 \Delta_i = 2^i$ with high probability. This implies that $\Delta_K = N$, contradicting the assumption that $B_K > 0$.
\end{proof}

\begin{lemma}\label{lemma:majority-k1}
The random variable $K_1$ satisfies $K_1 \in O(\log N)$ with high probability.
\end{lemma}
\begin{proof}
  Note that $K_0 = K_1$ implies that the claim holds by \lemmaref{lemma:k0-bound}.
  Hence, assume that $K_1 > K_0$ holds and $K_0$ is fixed. For $K_0 \le i < K_1$, let $U_i$ be the set of empty tokens at the end of iteration $i$. Let $Y_i(u)$ be an indicator variable for the event that $u \in U_i$ gets cancelled during the \emph{first} 
  round of iteration $i+1$. Since $\Pr[Y_i(u) = 1 \mid C_i = c] = 1 - c/N$ holds, the expected number of empty tokens after iteration $i+1$ satisfies
  \begin{align*}
    \E[C_{i+1} \mid C_i = c] &\le c - \E\left[\sum_{u \in U_i} Y_i(u) \right] = c - \sum_{u \in U} \E[Y_i(u)] = c - c(1-c/N) = c^2/N.
  \end{align*}
  By law of total expectation,
  \[
\E[C_{i+1}] = \E[C_{i+1} \mid C_{K_0}] = g(i+1-K_0),
  \]
  where $g$ is the recurrence from \lemmaref{lemma:recurrence} with $x = C_{K_0}$ and $n=N$. For $t = \log N + \log (\lambda+1) + \log \ln N$, the second claim of \lemmaref{lemma:recurrence} yields
  \[
  \Pr\left[C_{t+K_0} \ge 1\right] \le \E[C_{t+K_0}] \le g(t) \le 1/N^\lambda \le 1/n^\lambda. \qedhere
  \]
\end{proof}

\MAJalg*
\begin{proof}
  By \lemmaref{lemma:majority-k1}, the system reaches a configuration, where all tokens take the majority value within $O(\log N)$ iterations with high probability. As each iteration takes $O(\log N) = O(\log n)$ rounds, the algorithm stabilises in $O(\log^2 n)$ rounds with high probability.
  The state complexity is $t \in O(\log n)$, as the nodes count up to $t$ in each iteration and there are 5 token types.
\end{proof}

\end{document}